\documentclass[12pt]{article}

\usepackage[margin=1in]{geometry}
\usepackage[T1]{fontenc}
\usepackage[utf8]{inputenc}
\usepackage[protrusion=true,expansion=false]{microtype}
\usepackage{setspace}
\usepackage{graphicx}

\usepackage{amsmath,amssymb,mathtools}
\usepackage{amsthm}

\usepackage[authoryear,round]{natbib}
\setcitestyle{aysep={,},yysep={;}}

\usepackage{xcolor}
\definecolor{linkblue}{HTML}{1F5A94}
\usepackage[
  colorlinks=true,
  linkcolor=linkblue,
  citecolor=linkblue,
  urlcolor=linkblue
]{hyperref}

\newtheorem{proposition}{Proposition}
\newtheorem{lemma}[proposition]{Lemma}
\newtheorem{corollary}[proposition]{Corollary}
\theoremstyle{definition}

\theoremstyle{remark}

\newcommand{\E}{\mathbb{E}}
\newcommand{\ind}{\mathbf{1}}
\newcommand{\A}{\mathcal{A}}
\newcommand{\Iset}{\mathcal{I}}

\title{\textbf{Selling Information to a Competitor}}

\author{
  Thomas Falconer\(^{*}\) \qquad 
  Anubhav Ratha \\[0.5em]
  Jalal Kazempour \qquad 
  Pierre Pinson \qquad 
  Maryam Kamgarpour \\[0.5em]
  \small \(^{\star}\)Corresponding author: \href{mailto:falco@dtu.dk}{falco@dtu.dk}
}

\author{
  Thomas Falconer%
  \thanks{Corresponding author: \href{mailto:falco@dtu.dk}{falco@dtu.dk}}%
  \textsuperscript{ ,}%
  \thanks{Section for Energy Markets and Analytics, DTU, Denmark.}
  \qquad
  Anubhav Ratha%
  \thanks{Copenhagen Energy Trading, Denmark.}
  \\[0.5em]
  Jalal Kazempour\footnotemark[1]
  \qquad
  Pierre Pinson%
  \thanks{Dyson School of Design Engineering, Imperial College London, United Kingdom.}
  \qquad
  Maryam Kamgarpour%
  \thanks{Lab for Systems Control and Multiagent Optimization Research, EPFL, Switzerland.}
}

\date{\today}

\begin{document}

\maketitle

\begin{abstract}
Selling information to a competitor can generate revenue but may also help the rival choose better actions in the market they compete in. We study this trade-off when the buyer's demand for supplemental information depends on its own private information. The buyer selects from a menu of information products offered by the seller, who incurs a competitive loss proportional to the increase in the buyer's probability of choosing the correct downstream action. In a common-prior model with binary states and actions, we characterize profit-maximizing menus for two buyer types and for a continuum of types. We show that selling information can remain profitable despite the competitive loss, but withholding is optimal when this loss exceeds the buyer's valuation of the information. Stronger competition weakly reduces the profitability of selling information and the buyer's overall probability of choosing correctly downstream, though it need not make individual information products less informative.
\end{abstract}


\section{Introduction}
\label{sec:introduction}
%Information may therefore be particularly valuable to firms from which its holder would prefer to withhold it.
A firm's data can be especially valuable to its competitors. A vertically integrated platform, for example, may hold demand data that would help a merchant choose which products to stock, even though that merchant competes with the platform's own retail business. Selling information creates revenue, but better buyer decisions can reduce the seller's profit in the downstream market where that information is used. We ask how competition with the buyer changes what information a firm should sell, if any, and at what price.

This question belongs to the literature on markets for information, which studies how information is packaged and sold for use in decision making \citep{bergemann2019information}. We distinguish the seller's underlying data from the information a buyer learns from it. In our setting, the seller observes a market state that affects both firms' payoffs. An information product specifies what the buyer can learn about this state. It generates a message about the state, called a signal, which the buyer combines with its private information to choose an action. A fully informative product reveals the state with certainty, whereas a partially informative product leaves some uncertainty. In the theoretical literature, such a product is represented by a \textit{statistical experiment}: a probability distribution over signals the buyer may receive in each state \citep{blackwell1951comparison, blackwell1953equivalent}. The seller offers a menu, with each item consisting of an experiment and its corresponding price.

Our closest benchmark is \citet{bergemann2018design}, who study a privately informed buyer purchasing supplemental information from a monopolistic seller. The buyer's private information governs its willingness to pay for each experiment. To maximize revenue, the seller designs products and prices that induce different buyer types to select different menu items, a process known as screening. Their analysis shows why information products can differ in more than overall precision with respect to the underlying data: an experiment may change one type's downstream action while leaving another's unchanged. We retain their contracting framework but introduce a competitive loss when the buyer makes a better decision. The seller therefore maximizes revenue net of this loss.
More general contracting frameworks allow disclosure and payments to occur over several stages. In \citet{babaioff2012optimal}, both parties may hold private information, and partial disclosure can be interspersed with payments. We restrict attention to a menu of statistical experiments posted \textit{ex ante}, with each experiment purchased at a single deterministic price that does not depend on the realized state, signal, or action.

A related literature studies information sales to multiple buyers, where informing one buyer can affect the payoffs of others \citep[e.g.,][]{bimpikis2019information,agarwal2024towards}. In \citet{bonatti2024selling}, signals take the form of recommendations about which downstream action to choose. The authors show that competition makes receiving the correct recommendation exclusively more valuable. However, buyers must be willing to follow the seller's recommendations, which limits the degree of exclusivity. Their analysis also illustrates how stronger competition can increase the seller's revenue from information sales. \citet{rodriguez2024strategic} examines how strategic incentives and correlated private information shape the optimal sale of information. In these settings, the seller must account for interactions amongst buyers. Our seller instead competes directly with a single buyer and bears the cost of improving that buyer's decisions. Its problem is to design menu items for the different private beliefs that this buyer may hold.

Information sellers who also participate in a downstream market have earlier precedents, including \citet{admati1988selling}. \citet{gradwohl2023selling} study a firm that sells data about selected consumers to a rival and then competes with it in prices. \citet{castiglioni2023selling} study an interested seller facing a buyer whose private type determines its utility and budget, with contracts allowing refunds contingent on the buyer's action. In our setting, the improvement that makes information valuable to the buyer also makes its sale costly to the seller. We explore how this link changes the design and pricing of experiments when the buyer's existing information is private and payments cannot depend on subsequent signals or actions.

A complementary literature designs data markets for prediction tasks and determines how data providers should be compensated. \citet{agarwal2019marketplace} propose a marketplace in which the usefulness of data depends on the prediction task and on the other data available. Related work develops mechanisms for classification \citep{koutsopoulos2015auctioning} and regression \citep{pinson2022regression}, and addresses correlated-data trading \citep{feng2021uncovering}, strategic replication \citep{falconer2025towards}, and financial risk \citep{falconer2024bayesian}. We take the seller's underlying data as given and study what it should reveal when improving the buyer's decision also reduces its own downstream profit.

Related operational models study information sharing and analytics services alongside other channel decisions. \citet{lee2026analysis} examine the sharing and exchange of demand signals jointly with commissions, retail prices, and transfers, while \citet{xing2026effects} study analytics investment and information-service choice under reselling and agency modes. These models provide application-specific equilibrium and managerial insights. We focus on a joint information- and mechanism-design problem: choosing what information products reveal and setting prices that induce privately informed buyers to select amongst them. Our characterizations identify how competitive exposure and the distribution of buyer beliefs shape the information offered, which types purchase, and when withholding is optimal.

We consider two possible states and two actions. For example, demand may be high or low, and a firm may choose between high and low output. In the model, choosing the action that matches the state earns a payoff of one, while the other action earns zero. The seller observes the state perfectly. The buyer receives a private signal and forms a probability that the state is high. This probability is its type. A buyer whose belief is close to one-half has more to gain from full information than one already confident about either state. The firms share a common prior and agree on how the buyer's signal relates to the state. They thus use the same underlying probability model, although they observe different information.

Before the state and the buyer's private signal are realized, the seller commits to a menu of experiments and prices. The buyer then observes its private signal and either purchases one experiment or chooses its outside option, that is, not to purchase any. A purchased experiment generates a supplemental signal about the state. The buyer combines this signal with its existing information and chooses its downstream action. The price is fixed upfront and cannot depend on the state, the supplemental signal, or the action taken. Commitment prevents the seller from changing either the available products or their prices after learning the state.

Let \(\tau\) denote the seller's competitive exposure: its loss when the buyer chooses the correct action, measured relative to the buyer's unit benefit from that choice. This additive payoff specification isolates the harm from better buyer decisions. It does not model endogenous price or quantity competition. Under the common prior, an experiment that increases a particular type's probability of choosing correctly by \(\varepsilon\) creates a buyer value of \(\varepsilon\) and an expected competitive cost of \(\tau\varepsilon\). At price \(t\), the seller's expected profit from that type, relative to withholding information, is \(t-\tau\varepsilon\).
The buyer will not pay more than \(\varepsilon\), so
\[
    t-\tau\varepsilon \leq (1-\tau)\varepsilon.
\]
Withholding is therefore optimal when \(\tau\) is at least one. At equality, sales that extract the buyer's entire value can tie withholding. Below this threshold, information can create positive joint surplus, although the seller must still account for how one product affects the buyer's willingness to select another. At \(\tau=0\), the model reduces to the binary equal-payoff case of \citet{bergemann2018design}.

A general implication is that stronger competition weakly reduces the seller's optimal profit relative to withholding and the buyer's overall probability of choosing correctly. These conclusions hold for a fixed distribution of buyer beliefs without requiring the regularity conditions used below. They concern average outcomes across buyer types, meaning that competition may change which types purchase information or which products they select, without making every individual product less informative.

For two buyer types, the more uncertain type receives full information whenever competitive exposure is below the withholding threshold. The offering to the more certain type depends on the actions the two types would choose without additional information. If both would choose the same action, an optimal menu either sells full information to both or sells it only to the more uncertain type. If they would choose different actions, an optimal menu either sells full information to both or offers the more certain type a partially informative experiment. Stronger competition favors withholding from the more certain type in the first case and offering it a distinct product in the second. This partial experiment can be chosen independently of competitive exposure; competition changes whether the seller offers it.

For a continuum of buyer types, we assume that the distribution of buyer beliefs has a continuous, strictly positive density on the unit interval.
Below the withholding threshold, an optimal menu offers full information at a single price under a regularity condition on how the seller's marginal returns vary with the buyer's belief. This condition depends jointly on the belief distribution and competitive exposure. Buyers who remain sufficiently uncertain purchase; buyers already sufficiently confident about either state do not. When regularity holds at the competition levels compared, stronger competition raises the price and reduces the set of purchasers. Every purchaser still receives full information.

When regularity fails, we construct profit-maximizing menus that may include partially informative products, whilst ensuring that each type prefers its assigned product to the alternatives. Methodologically, our continuum analysis combines the experiment representation of \citet{bergemann2018design} with the virtual-surplus approach of \citet{myerson1981optimal}, which accounts for the surplus that must be left to privately informed buyers. For irregular distributions, the construction uses ironing---averaging marginal returns where necessary to satisfy buyer incentive constraints---and is related to the methods of \citet{toikka2011ironing}. We also show that competition can simplify the optimal menu. For every continuously differentiable belief density that is strictly positive throughout the closed unit interval, full information at a single price becomes optimal when competitive exposure is sufficiently close to the withholding threshold from below. Competition can therefore eliminate the need for partially informative products before selling information ceases to be profitable.

The remainder of the paper is organized as follows. Section~\ref{sec:environment} presents the environment. Section~\ref{sec:mech-design} formulates the seller's problem and develops reductions that simplify the design of experiments. Sections~\ref{sec:two-types} and~\ref{sec:continuum} characterize profit-maximizing menus for two buyer types and a continuum of types, respectively. Finally, Section~\ref{sec:conclusion} summarizes the findings and discusses directions for future work.

\section{Environment}
\label{sec:environment}
Consider two firms that compete in a downstream market, where they face a binary payoff-relevant state \(X\in\{0,1\}\), with realization \(x\). One firm (the seller) observes the state once it is realized and can sell supplemental information about it to the other firm (the buyer), which first privately observes its own signal and forms beliefs about \(x\). We index the firms by \(i\in\{\mathrm{b},\mathrm{s}\}\), where \(\mathrm{b}\) and \(\mathrm{s}\) denote the buyer and seller, respectively.

Before either the state or the buyer's private signal is realized, the seller commits to a menu of statistical experiments with corresponding prices, as detailed in Section~\ref{sec:mech-design}. After the seller observes \(x\) and the buyer observes its private signal, the buyer purchases at most one experiment from the menu, which generates a supplemental signal. After observing this signal, the buyer updates its beliefs, and finally both firms choose their respective actions \(a_i\in\A_i=\{0,1\}\). The timing is stated in full at the end of the section.

\paragraph{Information structure}
Let \(R\) denote the buyer's private signal, with realization \(r\). We assume the firms share a common prior over \(X\) and agree on the distribution of \(R\) conditional on \(X\). After \(r\) is realized, the buyer forms the posterior belief
\begin{equation}
    \theta
    =
    \Pr(X=1\mid R=r)
    \in[0,1],
    \label{eq:type}
\end{equation}
which we refer to as the buyer's type. From the seller's \textit{ex
ante} perspective, \(\theta\) is distributed according to \(F\), with type space \(\Theta=\operatorname{supp}F\subseteq[0,1]\). Common-prior
consistency implies that
\(X\mid\theta\sim\operatorname{Bernoulli}(\theta)\). Equivalently, for
every Borel set \(A\subseteq\Theta\) and every \(x\in\{0,1\}\), the
induced joint distribution satisfies
\begin{equation*}
    \Pr(X=x,\theta\in A)
    =
    \int_A z^x(1-z)^{1-x}\,dF(z).
\end{equation*}
Since the seller commits to a menu before \(X\) and \(R\) are realized, it is characterized under this joint distribution.

\paragraph{Reduced-form competition}
Let \(\boldsymbol{a}=(a_{\mathrm{b}},a_{\mathrm{s}})\) denote the action profile, with
\(\A
=\A_{\mathrm{b}}\times\A_{\mathrm{s}}\).
Firm \(i\)'s \textit{ex post} utility is
\begin{equation*}
    u_i(\boldsymbol{a},x;\tau)
    =
    \ind\{a_i=x\}
    -
    \tau\ind\{a_j=x\},
    \qquad
    j\neq i,
\end{equation*}
where \(\tau\geq 0\) measures competitive exposure: the loss a firm incurs when its competitor chooses the action that matches the state, relative to its own unit benefit from choosing correctly.
Each firm benefits from matching its own action to the state, while its utility is reduced when its competitor does so.
If the buyer purchases information from the seller at a price \(t\geq0\), the firms' transfer-inclusive profits are
\begin{equation}
\begin{aligned}
    \pi_{\mathrm{b}}(\boldsymbol{a},x,t;\tau)
    &=
    u_{\mathrm{b}}(\boldsymbol{a},x;\tau)-t,\\
    \pi_{\mathrm{s}}(\boldsymbol{a},x,t;\tau)
    &=
    u_{\mathrm{s}}(\boldsymbol{a},x;\tau)+t.
\end{aligned}
\label{eq:total-profits}
\end{equation}

Both firms are risk neutral and maximize expected profit.
Actions are simultaneous, so the buyer cannot infer \(x\) from the seller's action before choosing; observing the equilibrium action \(a_{\mathrm{s}}=x\) would otherwise reveal the state.
\begin{lemma}[Profit reduction]
\label{lem:payoff-reduction}
Since the seller observes the realized state, its strictly dominant
action is \(a_{\mathrm{s}}=x\). Substituting this action into
\eqref{eq:total-profits} gives
\begin{equation} 
\begin{aligned} 
    \pi_{\mathrm{b}} &= \ind\{a_{\mathrm{b}}=x\}-\tau-t,\\ \pi_{\mathrm{s}} &= 1-\tau\ind\{a_{\mathrm{b}}=x\}+t. 
\end{aligned} 
\label{eq:reduced-profits} 
\end{equation} 
As the additive terms \(-\tau\) and \(1\) are constant across all feasible mechanisms, we can omit them and evaluate the profits of the buyer and seller as
\begin{equation}
    \ind\{a_{\mathrm{b}}=x\}-t
    \qquad\text{and}\qquad
    t-\tau\ind\{a_{\mathrm{b}}=x\},
    \label{eq:mechanism-relevant-payoffs}
\end{equation}
respectively.
\end{lemma} 
\begin{proof} 
The competitive-loss term is independent of \(a_{\mathrm{s}}\), so an informed seller strictly prefers \(a_{\mathrm{s}}=x\). Substitution gives \eqref{eq:reduced-profits}, and omitting the two mechanism-independent constants, \(-\tau\) and \(1\), gives \eqref{eq:mechanism-relevant-payoffs}.
\end{proof} 

%Thus, \(\tau\) is a reduced-form measure of the seller's exposure to buyer accuracy. It changes the seller's objective but neither the buyer's information demand nor the implementability constraints.

Henceforth, we write \(a=a_{\mathrm{b}}\) and evaluate the optimal mechanism using the normalized payoffs in \eqref{eq:mechanism-relevant-payoffs}. In the absence of supplemental information, a buyer of type \(\theta\) chooses 
\begin{equation*} 
a(\theta) = 
\begin{cases} 
    0, & \theta<1/2,\\ 
    1, & \theta\geq1/2,
\end{cases} 
\end{equation*} 
and the corresponding expected normalized utility is
\begin{equation*} 
    \bar{u}(\theta) = \max\{\theta,1-\theta\}. 
\end{equation*} 
We refer to \(a(\theta)\) as the buyer's default action. Nonparticipation entails no supplemental signal and no transfer payment, thereby yielding \(\bar{u}(\theta)\) in normalized terms, or \(\bar{u}(\theta)-\tau\) under the original profit specification.

\paragraph{Timing}

The interaction between the buyer and seller proceeds as follows:
\begin{enumerate}
    \item The common prior, the induced type distribution \(F\), and competitive exposure \(\tau\) are publicly known.

    \item The seller commits \textit{ex ante} to a menu of statistical experiments and corresponding prices, that is, before \(x\) and \(r\) are realized.

    \item The state \(x\) and private signal \(r\) are realized. The seller observes \(x\), while the buyer observes \(r\) and forms its type \(\theta\) according to \eqref{eq:type}.

    \item The buyer purchases an experiment from the menu, which generates a private supplemental signal conditional on \(x\).

    \item After observing the supplemental signal, the buyer chooses action \(a_{\mathrm{b}} = a\) and the seller simultaneously chooses \(a_{\mathrm{s}}=x\); \textit{ex post} profits are then realized.
\end{enumerate}
Because the menu is chosen before the seller observes \(x\), its
posting does not itself reveal information about the state. Moreover,
conditional on the state and the selected experiment, the supplemental
signal is independent of the private signal \(r\). Once posted, the
menu cannot be revised after the state is realized; after the buyer's
choice, neither the selected experiment nor its price can be revised.

%We assume an enforceable commitment technology. After the buyer selects an item \(E\), a verifiable randomization device draws \(S\) from \(\sigma(\cdot\mid X)\). The seller cannot suppress, alter, or condition the draw on information other than \(X\) and the selected item. This assumption is essential because an interested seller would generally prefer to manipulate disclosure after receiving the price and observing the state.

\section{Mechanism Design}
\label{sec:mech-design}

This section reduces the seller's choice over arbitrary menus to a tractable screening problem. We first define how buyers value experiments and state the participation and incentive constraints. We then show that signals inducing the same action for the intended type can be pooled, so each allocated experiment has at most two signals. Finally, under the common prior, the buyer's value of an experiment and the resulting competitive cost to the seller are both determined by how much it increases the probability that the buyer chooses the action matching the state. This identifies when selling information can be profitable. Below the withholding threshold, an optimal menu can be represented by responsive one-sided experiments, which recommend the correct action with probability one in at least one state.

\paragraph{Value of experiments}
\label{sec:general-experiments}
We represent each information product by a statistical experiment, which specifies the distribution of the supplemental signal conditional on the realized state.
Fix a common countable signal space \(\mathcal{S}\). An experiment \(E=(\mathcal{S},\sigma)\) is characterized by the likelihood function 
\begin{equation*}
\sigma(s\mid x) = \Pr(S=s\mid X=x,E),
\end{equation*}
where \(S\) is the supplemental-signal random variable,
\(\sigma(s\mid x)\geq0\), and
\(\sum_{s\in\mathcal{S}}\sigma(s\mid x)=1\) for every
\(x\in\{0,1\}\).

We restrict attention to finite-support experiments, such that \(\sigma(s\mid0)+\sigma(s\mid1)>0\) for only finitely many signals \(s\in\mathcal{S}\). Let \(\mathcal{E}\) denote this set of admissible experiments. After purchasing an experiment \(E\in\mathcal{E}\) and observing signal \(s\), a buyer of type \(\theta\) chooses
\begin{equation*}
    a(s\mid\theta,E)
    \in
    \arg\max_{a\in\{0,1\}}
    \left\{
        (1-\theta)\sigma(s\mid0)\ind\{a=0\}
        +
        \theta\sigma(s\mid1)\ind\{a=1\}
    \right\}.
    %label{eq:post-signal-action}
\end{equation*}
We suppress \(E\) when the experiment is clear from context. %Multiplying the buyer's expected utility conditional on each signal by the probability of that signal cancels the denominator in Bayes' rule. 
The incremental value of experiment \(E\) to type \(\theta\) is therefore
\begin{equation*}
    V(E,\theta)
    =
    \sum_{s\in\mathcal{S}}
    \max
    \left\{
        (1-\theta)\sigma(s\mid0),
        \theta\sigma(s\mid1)
    \right\}
    -
    \bar{u}(\theta),
\end{equation*}
with $V(E,\theta) \geq 0$ since the buyer can ignore \(s\) and choose
its default action \(a(\theta)\).

\paragraph{Direct menus}
\label{sec:direct-menus}
A menu consists of experiment--price pairs \((E,t)\),
where \(E\in\mathcal{E}\) and \(t\geq0\).
Prices may differ between menu items but cannot depend on the realized state \(x\), supplemental signal \(s\), or subsequent actions \(\boldsymbol{a}\). After observing its type, the buyer either purchases an experiment or chooses its outside option (i.e., no purchase).
As the outside option yields zero incremental utility, it can be represented by the zero-price uninformative experiment
\(E^\circ=(\mathcal{S},\sigma^\circ)\), where for some
\(s^\circ\in\mathcal{S}\),
\begin{equation*}
    \sigma^\circ(s^\circ\mid x)=1, 
\end{equation*}
for every \(x\in\{0,1\}\).
The \emph{withholding menu} offers only \((E^\circ,0)\).

With the revelation principle, we may restrict attention to direct
menus. A direct menu consists of a family of allocated experiments
\((E_{\hat\theta})_{\hat\theta\in\Theta}\) and a transfer function
\(t:\Theta\to\mathbb{R}_{\geq0}\), represented as
\begin{equation*}
    \mathcal{M}
    =
    \left\{
        \bigl(E_{\hat{\theta}},t(\hat{\theta})\bigr):
        \hat{\theta}\in\Theta
    \right\}.
\end{equation*}
Writing
\(E_{\hat{\theta}}=(\mathcal{S},\sigma_{\hat{\theta}})\), we require
\(t(\hat{\theta})\) and \(\sigma_{\hat{\theta}}(s\mid x)\) to be Borel
measurable in \(\hat{\theta}\) for every \(s\in\mathcal{S}\) and
\(x\in\{0,1\}\).

The truthful net utility, or information rent, of type \(\theta\) is
\begin{equation*}
    U(\theta)
    =
    V(E_{\theta},\theta)-t(\theta).
\end{equation*}

As \(V(E,\theta)\) evaluates each experiment under the buyer's optimal action, deviations are already incorporated. Feasibility therefore requires only incentive compatibility across all possible reports, individual rationality, and nonnegative prices.
Therefore, a direct menu is feasible if, for every
\(\theta,\hat{\theta}\in\Theta\),
\begin{equation*}
    U(\theta)\geq V(E_{\hat{\theta}},\theta)-t(\hat{\theta}),
    \qquad U(\theta)\geq 0,
    \qquad t(\theta)\geq 0,
\end{equation*}
which ensure that truthful reporting and participation are optimal.
%These constraints ensure that truthful reporting and participation are optimal, though possibly only weakly so. %Throughout, we evaluate a feasible direct menu at its designated truthful outcome. 
%We adopt favorable selection amongst the buyer's best responses.
%Specifically, we assume that a buyer participates when indifferent to its outside option, chooses the menu item intended for its type when several items give it the same utility, and follows any recommendation that occurs with positive probability under its allocated experiment whenever doing so is weakly optimal.
%This is an equilibrium-selection convention rather than a consequence of incentive compatibility.

Decision accuracy is the probability that the buyer chooses the action
matching the state. The value \(V(E,\theta)\) is the improvement in
decision accuracy generated by experiment \(E\). Hence,
\[
    0 \leq V(E,\theta)
      \leq 1-\bar{u}(\theta)
      \leq \frac{1}{2}.
\]
For a feasible truthful allocation, individual rationality and
nonnegative prices imply
\(0\leq t(\theta),U(\theta)\leq V(E_\theta,\theta)\).
As such, all objective terms below are bounded and integrable.
%Under the common prior, the seller's expected incremental profit from
%type \(\theta\) is
%
%\begin{equation*}
%t(\theta)-\tau V(E_{\theta},\theta).
%\end{equation*}
%

\paragraph{Responsive experiments}
\label{sec:responsive-reduction}

An experiment allocated to type \(\theta\) is \emph{responsive} if each signal recommends a different action and type \(\theta\) finds it optimal to follow every recommendation.
A direct menu is responsive if every experiment is responsive for its intended type.

The following result adapts the binary case of Proposition~1 in
\citet{bergemann2018design} to an interested seller.

\begin{proposition}[Responsive reduction]
\label{prop:responsive}
For every feasible direct menu of finite-support experiments, there exists a responsive analogue that preserves every intended type's value and transfer, and hence expected seller profit. 
\end{proposition}

Pooling signals that induce the same action for the intended type preserves that type's probability of matching the state. Under the
common prior, the seller's competitive cost relative to withholding is \(\tau V(E_\theta,\theta)\), so preserving the intended type's value also preserves this cost. Since transfers remain unchanged, the
garbling preserves both information-sales revenue and the seller's downstream payoff.

We may therefore restrict attention to menus in which each allocated
experiment has at most two nonnull signals, \(s_0\) and \(s_1\),
interpreted as recommendations of actions \(0\) and \(1\),
respectively. If we write
\begin{equation*}
    q_0=\sigma(s_0\mid0),
    \qquad
    q_1=\sigma(s_1\mid1),
    \qquad
    q_0,q_1\in[0,1],
\end{equation*}
the experiment can then be represented by
\begin{equation*}
    E(q_0,q_1)
    =
    \begin{pmatrix}
        q_0 & 1-q_0\\
        1-q_1 & q_1
    \end{pmatrix},
\end{equation*}
where rows index states and columns index signals. Thus, full information is \(E(1,1)\), and the experiments \(E(1,0)\) and \(E(0,1)\) are uninformative and always recommend actions \(0\) and \(1\), respectively.

Responsiveness for the intended type means that each signal recommends an optimal action, or equivalently, type \(\theta\) obeys the recommendations if and only if
\begin{align}
    (1-\theta)q_0
    &\geq
    \theta(1-q_1),
    \label{eq:obedience-zero}\\
    \theta q_1
    &\geq
    (1-\theta)(1-q_0).
    \label{eq:obedience-one}
\end{align}

If the buyer follows the recommendations, it matches the state with
probability \((1-\theta)q_0+\theta q_1\). Its incremental value relative
to the no-information payoff \(\bar{u}(\theta)\) is therefore
\begin{equation*}
    \Delta(E,\theta)
    =
    (1-\theta)q_0+\theta q_1-\bar{u}(\theta).
\end{equation*}

Of course, a buyer who purchases an experiment intended for another type need not
follow its recommendations. With two signals and two actions, the buyer
can follow, reverse, or ignore the signal and choose a constant action.
The best constant response is its default action, which yields zero
incremental value, while following yields \(\Delta(E,\theta)\). For a
responsive experiment, obedience conditions
\eqref{eq:obedience-zero} and \eqref{eq:obedience-one} imply
\(q_0+q_1\geq1\), meaning that each signal is weakly more likely in its corresponding state.\footnote{For \(\theta\in(0,1)\), multiply the two obedience inequalities and cancel \(\theta(1-\theta)\) to obtain \(q_0q_1\geq(1-q_0)(1-q_1)\), which is equivalent to \(q_0+q_1\geq1\). At \(\theta=0\), obedience after \(s_1\) requires \(q_0=1\); at \(\theta=1\), obedience after \(s_0\) requires \(q_1=1\).} Reversing therefore cannot outperform the default. The buyer's choice is thus between following and ignoring the recommendations, as formalized in the following lemma.

\begin{lemma}[Deviation value]
\label{lem:deviation}
For any binary recommendation experiment satisfying
\(q_0+q_1\geq1\), its incremental value to type \(\theta\) is
\begin{equation}
    V(E,\theta)
    =
    \max\{\Delta(E,\theta),0\}.
    \label{eq:deviation-value}
\end{equation}
\end{lemma}

This applies to every experiment in a responsive menu.
For the intended type, responsiveness additionally guarantees that following
the recommendations is optimal. Hence
\begin{equation*}
    V(E_{\theta},\theta)
    =
    \Delta(E_{\theta},\theta).
\end{equation*}
Under the common prior, \(\Delta(E_{\theta},\theta)\) also equals the increase in
the seller-assessed probability that the buyer matches the state. As a correct
buyer action reduces the seller's payoff by \(\tau\), the seller's incremental
competitive cost relative to withholding is $ \tau\Delta(E_{\theta},\theta)$.

With the direct-mechanism formulation above and Proposition~\ref{prop:responsive}, the
seller chooses a feasible responsive direct menu
$\mathcal M = \{E_\theta,t(\theta)\}_{\theta\in\Theta}$ to maximize expected incremental profit relative to withholding.
Denote this objective by $\Pi(\mathcal M)$. Since $V(E_\theta,\theta)=\Delta(E_\theta,\theta)$ for every intended type, the seller's problem is
\begin{equation*}
    \sup_{\{E_\theta\}_{\theta\in\Theta},\,t}
    \Pi(\mathcal M)
    :=\int_{\Theta}
    \left[
        t(\theta)
        -
        \tau\Delta(E_{\theta},\theta)
    \right]
    dF(\theta),
\end{equation*}
over experiments and transfers constituting a feasible responsive direct menu $\mathcal M$.

\paragraph{Competition threshold}
\label{sec:competition-threshold}
Writing \(t(\theta)=V(E_\theta,\theta)-U(\theta)\), the seller's
typewise incremental-profit contribution is
\begin{equation*}
    t(\theta)-\tau V(E_\theta,\theta)
    = (1-\tau)V(E_\theta,\theta)-U(\theta).
\end{equation*}
Because both \(V\) and \(U\) are nonnegative, this contribution is
never positive when \(\tau\geq1\). For \(\tau>1\), strict expected loss
requires positive information value on a set of positive \(F\)-measure;
changing only null types does not affect the expected objective.

\begin{proposition}[Withholding at \(\tau\geq1\)]
\label{prop:no-trade}
If \(\tau\geq1\), withholding is optimal. If \(\tau>1\), every feasible
menu that supplies strictly positive incremental value on a set of
types with positive \(F\)-measure is strictly dominated by withholding.
At \(\tau=1\), a feasible menu is optimal if and only if the buyer's expected information rent is zero:
\[
  \int_\Theta U(\theta)\,dF(\theta)=0.
\]
In particular, every feasible menu that charges each intended type its full incremental value ties withholding.
\end{proposition}

%At \(\tau=1\), withholding is optimal but may not be uniquely optimal. Any feasible menu satisfying \(t(\theta)=V(E_{\theta},\theta)\) also yields zero expected incremental profit and therefore ties withholding.
\begin{corollary}[Comparative statics of competitive exposure]
\label{cor:competitive-exposure}
Fix the distribution of buyer beliefs and all primitives other than
competitive exposure. For a feasible menu \(\mathcal{M}\), define
\[
    \overline V(\mathcal{M})
    := \int_{\Theta} V(E_\theta,\theta)\,dF(\theta).
\]
Let \(\mathcal{M}_1\) and \(\mathcal{M}_2\) be optimal menus at
\(0\leq\tau_1<\tau_2\), respectively. Then
\[
    \overline V(\mathcal{M}_2)
    \leq \overline V(\mathcal{M}_1).
\]
The seller's optimal expected profit relative to withholding is also
weakly decreasing in competitive exposure.
\end{corollary}

\begin{proof}
The feasible menu set is independent of \(\tau\), and expected
incremental profit can be written as
\[
    \Pi(\mathcal{M};\tau)
    = \int_{\Theta} t(\theta)\,dF(\theta)
      - \tau\overline V(\mathcal{M}).
\]
Optimality of \(\mathcal{M}_1\) and \(\mathcal{M}_2\) gives
\[
    \Pi(\mathcal{M}_1;\tau_1)\geq\Pi(\mathcal{M}_2;\tau_1),
    \qquad
    \Pi(\mathcal{M}_2;\tau_2)\geq\Pi(\mathcal{M}_1;\tau_2).
\]
Adding these inequalities yields
\[
    (\tau_2-\tau_1)
    \bigl[\overline V(\mathcal{M}_1)-\overline V(\mathcal{M}_2)\bigr]\geq0,
\]
proving the first claim. Since \(\overline V(\mathcal{M})\geq0\), the profit
from every fixed feasible menu is weakly decreasing in \(\tau\).
Its supremum over feasible menus is therefore also weakly decreasing.
\end{proof}

The buyer's overall probability of choosing correctly equals
\[
    \int_{\Theta}\overline u(\theta)\,dF(\theta)+\overline V(\mathcal{M}).
\]
The first term is fixed when the belief distribution is fixed.
Corollary~\ref{cor:competitive-exposure} thus also implies that
the buyer's overall probability of choosing correctly weakly
decreases with competitive exposure. This is an average across
types and does not imply that every individual experiment becomes
less informative.

\paragraph{One-sided experiments}
\label{sec:one-sided-experiments}

The menu characterizations below focus on \(0\leq\tau<1\),
where an additional unit of probability that the buyer's action matches the state creates positive joint surplus. A responsive binary experiment is \emph{one-sided} if
\begin{equation*}
    \max\{q_0,q_1\}=1, 
\end{equation*}
so the correct action is recommended with probability one in at least
one state. Equivalently, one of its off-diagonal probabilities is zero.
If the experiment is informative, it contains at least one signal that
occurs in only one state and perfectly reveals that state whenever the
signal is drawn.

At \(\tau=0\), the next result is the binary symmetric-payoff
case of
\citet{bergemann2018design}, who call such experiments \textit{concentrated}.
Our extension identifies the gain from concentration as
\((1-\tau)\varepsilon\) and hence yields the competition threshold.

\begin{proposition}[Reduction to one-sided experiments]
\label{prop:concentration}
For a feasible responsive binary menu, let
\begin{equation*}
E_{\theta}=E(q_0(\theta),q_1(\theta))
\qquad\text{and}\qquad
\varepsilon(\theta)
=
\min \{1-q_0(\theta),1-q_1(\theta) \},
\end{equation*}
where \(\varepsilon(\theta)\) is the smallest common increment to the
two diagonal probabilities that makes at least one of them equal to one.
Replace each menu item by
\begin{equation*}
\bigl(E'_\theta,t'(\theta)\bigr)
=
\left(
E\bigl(
q_0(\theta)+\varepsilon(\theta),
q_1(\theta)+\varepsilon(\theta)
\bigr),
t(\theta)+\varepsilon(\theta)
\right).
\end{equation*}
The transformed menu is feasible and one-sided at every report, and
the change in the seller's expected incremental profit is
\begin{equation}
  (1-\tau)
  \int_\Theta\varepsilon(\theta)\,dF(\theta).
  \label{eq:concentration-gain}
\end{equation}
Consequently, if \(\tau<1\), every optimal responsive binary menu is
one-sided \(F\)-almost everywhere. Every such optimum has a
feasible, profit-equivalent representative that is one-sided at every
report. If \(\tau=1\), every feasible responsive binary menu has a
profit-equivalent one-sided counterpart.
\end{proposition}

Figure~\ref{fig:experiments} illustrates this argument.
At $\theta=1/2$, replacing $E(4/5,3/10)$ by $E(1,1/2)$ raises the
buyer's value from $1/20$ to $1/4$. Increasing the price by $1/5$
preserves this type's utility and raises the seller's incremental
profit by $(1-\tau)/5$. The gain is positive for $\tau<1$, since
the additional payment exceeds the competitive cost increment $\tau/5$.

Panel~(b) compares the value of the resulting experiment to each buyer type with one of full information. Full information is weakly most valuable for
every type. Its value is greatest
for the most uncertain type, $\theta=1/2$, and vanishes at
$\theta\in\{0,1\}$. For a partially informative experiment, however,
value also depends on the direction of the buyer's belief.
For example, under $E(1,1/2)$, types $\theta=1/3$ and $\theta=2/3$
are equally distant from indifference but the value of this experiment is asymmetric. Thus $|\theta-1/2|$ alone does not determine willingness to pay for partial information.

\begin{figure}[t]
    \centering
    \includegraphics[width=\textwidth]{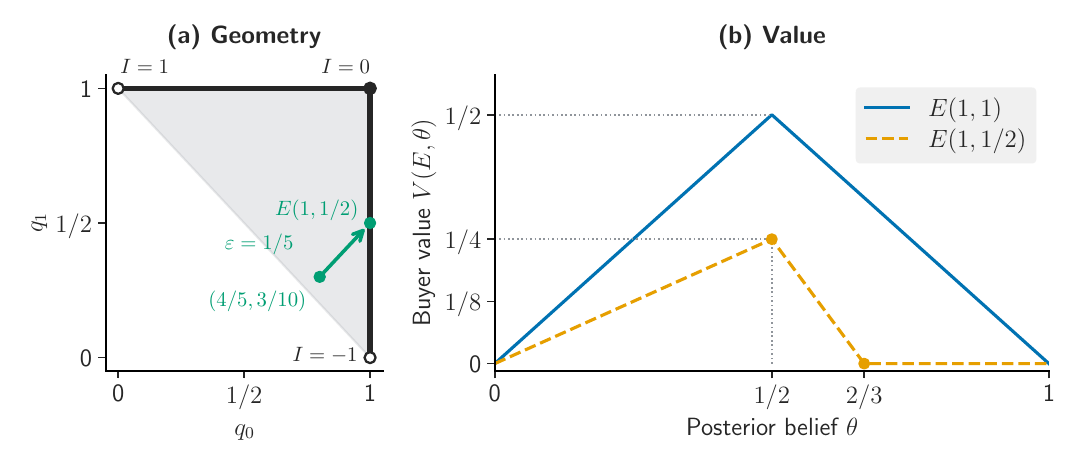}
    \caption{Geometry and value of experiments.
    In (a), the arrow raises both diagonal probabilities by $1/5$, moving
    $E(4/5, 3/10)$ to the one-sided boundary at $E(1,1/2)$.
    The shaded region satisfies $q_0+q_1\geq1$.
    Boundary labels denote the signed-distortion coordinate
    $I=q_1-q_0$ introduced in Section~\ref{sec:continuum}: $I=0$
    denotes full information, and $I=\pm1$ denote uninformative experiments.
    (b) Buyer value under full information (solid) and $E(1,1/2)$
    (dashed), which varies with the buyer's type.
}
    \label{fig:experiments}
    \label{fig:signed-distortion}
\end{figure}

Combining Propositions~\ref{prop:no-trade} and
\ref{prop:concentration}, restricting attention to responsive one-sided experiments leaves the seller's supremum unchanged for every \(\tau\geq0\). For \(\tau<1\), this concentration strictly improves any menu that is non-one-sided on a set of positive \(F\)-measure and transforms the remaining null reports without changing expected profit. At \(\tau=1\), the transformation is profit neutral. For \(\tau>1\), the conclusion follows from optimal withholding. In both cases \(\tau\geq1\), the supremum is attained by the following one-sided representation of withholding:
\[
  \bigl(E_{\theta},t(\theta)\bigr)
  =
  \begin{cases}
    \bigl(E(1,0),0\bigr), & \theta<1/2,\\
    \bigl(E(0,1),0\bigr), & \theta\geq1/2.
  \end{cases}
\]
Both experiments are uninformative, one-sided, and obedient for their
intended types. Thus, the economically substantive one-sidedness result
is the restriction on optimal information products when \(\tau<1\);
for \(\tau\geq1\), one-sidedness follows trivially from optimal
withholding.

\section{Two posterior types}
\label{sec:two-types}

The two-type case illustrates the seller's screening problem:
offering information to a buyer who values it less can limit the price
charged to a buyer who values it more. Whether the seller can avoid
this depends on the actions the two types would take without
supplemental information.

Consider two nonindifferent interior types with unequal values for full information and let \(0\leq\tau<1\). Type \(h\) has the higher value for full information and occurs with probability \(\gamma\in(0,1)\); type \(\ell\) has the lower value and occurs with probability \(1-\gamma\). The labels \(h\) and \(\ell\) therefore rank information values, not posterior beliefs.
Relabel the states so that \(\theta_h>1/2\). We study two configurations:
\begin{align}
  \text{same default:}\quad &\frac12<\theta_h<\theta_\ell<1,
  \label{eq:same-default}\\
  \text{opposite defaults:}\quad &\theta_\ell<\frac12<\theta_h,
  \qquad 1-\theta_h>\theta_\ell,
  \label{eq:opposite-default}
\end{align}
Thus, \(a(\theta_\ell)=a(\theta_h)=1\) in the first case, whereas
\(a(\theta_\ell)=0\) and \(a(\theta_h)=1\) in the second.

Full information allows the buyer to match the state with probability
one, so its incremental value to type \(\theta\) is
\[
  1-\bar u(\theta)=\min\{\theta,1-\theta\}.
\]
Define
\begin{equation*}
  H=1-\theta_h,
  \qquad
  L=\min\{\theta_\ell,1-\theta_\ell\}.
\end{equation*}
The restrictions in \eqref{eq:same-default} and
\eqref{eq:opposite-default} imply \(H>L\), so type \(h\) is closer to
indifference and values full information more.

At \(\tau=0\), the analysis below reduces to the binary two-type
characterization of \citet{bergemann2018design}; their congruent and
noncongruent cases correspond to our same- and opposite-default
configurations. Our extension reveals that full information at the top
and the canonical separating product survive for \(\tau<1\), while more intense competition lowers the regime cutoffs and thereby expands exclusion or separation.

%The menus below are weakly implemented under the favorable-selection convention stated above. In particular, type \(\ell\)'s participation constraint binds, and the high-type incentive constraint binds in the same-default menus and in the opposite-default menus for \(\alpha\geq\bar\alpha\), as well as at the degenerate endpoint \(\alpha=0\). Thus, ``separation'' means weak incentive-compatible separation. Without favorable selection, the displayed profit levels are strict-implementation suprema approached through arbitrarily small price perturbations.

%Several menus below leave a buyer indifferent between its allocated item and another available option, including nonparticipation. Under the selection convention stated above, the displayed profit levels are attained. Without that convention, they should generally be interpreted as strict-implementation suprema approached through arbitrarily small price perturbations.

For each \(k\in\{\ell,h\}\), write \(E_k:=E_{\theta_k}\) and
\(t_k:=t(\theta_k)\), and let
\[
  U_k=V(E_k,\theta_k)-t_k
\]
denote type \(k\)'s information rent. Rearranging, the seller's expected
incremental profit can therefore be written as
\begin{equation*}
\Pi(\mathcal M) = \gamma\Bigl[(1-\tau)V(E_h,\theta_h)-U_h\Bigr] + (1-\gamma)\Bigl[(1-\tau)V(E_\ell,\theta_\ell)-U_\ell\Bigr].
\end{equation*}

Since \(1-\tau>0\), increasing a type's incremental value \(V(E_k,\theta_k)\) raises profit whenever information rents remain unchanged. Thus, any distortion in the information supplied to type \(\ell\) serves a screening purpose: giving type \(\ell\) more information may require the seller to leave additional rent to type \(h\).

When the types have the same default, both benefit from information that helps them recognize if \(X=0\). Any experiment useful to type \(\ell\) is therefore also useful to type \(h\). With opposite defaults, by contrast, type \(\ell\) values information that helps it recognize state \(X=1\), while type \(h\) already chooses this action by default. The seller may then be able to provide information to type \(\ell\) without making its product attractive to type \(h\).

To make this screening effect precise, let \(\alpha\in[0,1]\) denote the fraction of type \(\ell\)'s full-information value supplied by its allocated experiment. Therefore, the experiment gives type \(\ell\)
incremental value \(L\alpha\). By Proposition~\ref{prop:concentration}, it
is enough to consider one-sided experiments. Under the common prior,
all responsive experiments that give type \(\ell\) value \(L\alpha\) impose
the same competitive cost, \(\tau L\alpha\). Yet, they can differ in their value to type \(h\). The seller therefore chooses, amongst experiments delivering \(L\alpha\) to type \(\ell\), one that minimizes its
attractiveness to type \(h\).

\begin{lemma}[Least-attractive low-type experiment]
\label{lem:canonical}
Fix \(\alpha\in[0,1]\). Amongst the one-sided experiments that are obedient
for type \(\ell\) and satisfy
\[
  \Delta(E,\theta_\ell)=L\alpha,
\]
type \(h\)'s value is weakly minimized by
\begin{equation*}
E_\ell^\star(\alpha)
=
\begin{cases}
  E(\alpha,1), & \text{if the types have the same default},\\
  E(1,\alpha), & \text{if the types have opposite defaults}.
\end{cases}
\end{equation*}
Define the minimum rent that makes type \(h\)'s intended item weakly optimal relative to the low-type item priced at \(L\alpha\) by
\begin{equation}
  r_h(\alpha)
  =
  \max \left\{
    V\bigl(E_\ell^\star(\alpha),\theta_h\bigr)-L\alpha,
  0\right\}.
  \label{eq:required-rent}
\end{equation}
Then
\begin{equation}
  0\leq r_h(\alpha)\leq H-L.
  \label{eq:rent-bound}
\end{equation}
\end{lemma}

With the same default, \(E_\ell^\star(\alpha)=E(\alpha,1)\) always recommends action \(1\) when \(X=1\), and recommends action \(0\) with probability \(\alpha\) when \(X=0\). With opposite defaults, \(E_\ell^\star(\alpha)=E(1,\alpha)\) always recommends action \(0\) when \(X=0\), and recommends action \(1\) with probability \(\alpha\) when \(X=1\).
The latter experiment directs information toward the state that type \(\ell\) needs help recognizing but that type \(h\) already favors by default.

Suppose the low-type experiment is sold at \(t_\ell=L\alpha\), thereby extracting all of type \(\ell\)'s value. Type \(\ell\) then receives zero rent, and high-type incentive compatibility, together with individual
rationality, requires
\[
  U_h\geq r_h(\alpha),
\]
which is a lower bound on type \(h\)'s rent, not an identity for every feasible menu. The bound can be attained while setting \(U_\ell=0\): give type \(h\) full information and charge
\[
  t_h=H-r_h(\alpha).
\]
Type \(h\) then weakly prefers its own allocated experiment by the definition of
\(r_h(\alpha)\). Type \(\ell\)'s utility from purchasing the high-type experiment is
\[
  L-t_h=L-H+r_h(\alpha)\leq0,
\]
where the inequality follows from \eqref{eq:rent-bound}. Thus the construction is feasible. The next proposition shows that distorting the information allocated to type \(h\) cannot improve upon it.

\begin{proposition}[Full information for the higher-value type]
\label{prop:no-distortion}
Under the common prior and with \(0\leq\tau<1\), every optimal two-type menu
under either \eqref{eq:same-default} or
\eqref{eq:opposite-default} allocates type \(h\) an experiment equivalent
to full information.
\end{proposition}

Proposition~\ref{prop:no-distortion} and Lemma~\ref{lem:canonical} allow the seller to restrict attention to a one-parameter family of menus. For \(\alpha\in[0,1]\), let \(\mathcal M(\alpha)\) allocate \(E_\ell^\star(\alpha)\) to type \(\ell\), full information to type \(h\), and set
\[
  t_\ell(\alpha)=L\alpha,
  \qquad
  t_h(\alpha)=H-r_h(\alpha).
\]
Type \(\ell\) receives zero rent and type \(h\) receives rent \(r_h(\alpha)\). The seller's expected incremental profit is
\begin{equation}
  \Pi(\alpha)
  :=
  \Pi\bigl(\mathcal M(\alpha)\bigr)
  =
  \gamma\left[(1-\tau)H-r_h(\alpha)\right]
  +(1-\gamma)(1-\tau)L\alpha.
  \label{eq:two-type-reduced-profit}
\end{equation}
The seller's remaining problem is to maximize \(\Pi(\alpha)\),
thus the distinction between the two type configurations is entirely captured by how the required rent \(r_h(\alpha)\) varies with the information supplied to type \(\ell\).

\paragraph{Same-default types}

When both types initially choose action \(1\), the canonical low-type
experiment is \(E_\ell^\star(\alpha)=E(\alpha,1)\). Its values are
\[
  V\bigl(E_\ell^\star(\alpha),\theta_\ell\bigr)=L\alpha,
  \qquad
  V\bigl(E_\ell^\star(\alpha),\theta_h\bigr)=H\alpha.
\]
Consequently,
\[
  r_h(\alpha)
  =
  H\alpha-L\alpha
  =
  (H-L)\alpha
  =
  (\theta_\ell-\theta_h)\alpha.
\]
Thus, supplying information to type \(\ell\) immediately creates rent
for type \(h\).

\begin{proposition}[Optimal menu with the same default]
\label{prop:same-default}
If \eqref{eq:same-default} holds, there is an optimal menu in which type \(h\) receives full information and type \(\ell\) receives
\(E(\alpha,1)\), with transfers
\begin{equation}
  t_\ell(\alpha)=L\alpha,
  \qquad
  t_h(\alpha)=H-(\theta_\ell-\theta_h)\alpha.
  \label{eq:same-default-transfers}
\end{equation}
Expected incremental profit is
\begin{equation}
  \Pi_{\mathrm{same}}(\alpha)
  =
  \gamma(1-\tau)H
  +
  \alpha\left[
    (1-\gamma)(1-\tau)L
    -
    \gamma(\theta_\ell-\theta_h)
  \right].
  \label{eq:same-default-profit}
\end{equation}
Define
\begin{equation}
  \gamma_{\mathrm{same}}(\tau)
  =
  \frac{(1-\tau)L}{H-\tau L}
  =
  \frac{(1-\tau)(1-\theta_\ell)}
       {(1-\theta_h)-\tau(1-\theta_\ell)}.
  \label{eq:gamma-same}
\end{equation}
Because \(H>L\) and \(\tau<1\),
\(\gamma_{\mathrm{same}}(\tau)\in(0,1)\).
If \(\gamma<\gamma_{\mathrm{same}}(\tau)\), the seller sets \(\alpha=1\):
both types receive full information and pay \(L\). If
\(\gamma>\gamma_{\mathrm{same}}(\tau)\), the seller sets \(\alpha=0\):
type \(\ell\) receives no useful information and pays zero, while type
\(h\) receives full information and pays \(H\). Every
\(\alpha\in[0,1]\) is optimal at equality.
\end{proposition}

The coefficient on \(\alpha\) in \eqref{eq:same-default-profit} compares the expected marginal surplus from serving type \(\ell\), \((1-\gamma)(1-\tau)L\), with the expected additional rent left to type \(h\), \(\gamma(H-L)\). As both terms are constant in \(\alpha\),
the seller chooses either full-information pooling or exclusion. Increasing \(\tau\) reduces the surplus from serving type \(\ell\) without changing the rent cost. By
\eqref{eq:gamma-same}, the cutoff decreases with \(\tau\), and stronger competition expands the exclusion region.

\paragraph{Opposite-default types}

Suppose type \(h\) initially chooses action \(1\), while type \(\ell\) initially chooses action \(0\). Here \(L=\theta_\ell\), and the canonical low-type experiment is \(E_\ell^\star(\alpha)=E(1,\alpha)\), with values
\[
  V\bigl(E_\ell^\star(\alpha),\theta_\ell\bigr)=L\alpha,
  \qquad
  V\bigl(E_\ell^\star(\alpha),\theta_h\bigr)
  =
  \max\{\theta_h\alpha+1-2\theta_h,0\}.
\]
It follows from \eqref{eq:required-rent} that
\[
  r_h(\alpha)
  =
  \max\left\{
    (\theta_h-L)\alpha+1-2\theta_h,\,
    0
  \right\}.
\]

Define
\begin{equation}
  \bar \alpha
  =
  \frac{2\theta_h-1}{\theta_h-L}
  =
  \frac{2\theta_h-1}{\theta_h-\theta_\ell}.
  \label{eq:zbar}
\end{equation}
The condition \(H>L\) implies \(\bar \alpha \in(0,1)\). Hence, for \(\alpha\leq\bar \alpha\),
\[
  r_h(\alpha)=0,
\]
whereas for \(\alpha\geq\bar \alpha\),
\[
  r_h(\alpha)
  =
  (\theta_h-L)\alpha+1-2\theta_h.
\]

Type \(h\) may assign positive gross value to the low-type experiment before \(\alpha\) reaches \(\bar \alpha\), which does not create information rent because, for every \(\alpha\leq\bar \alpha\),
\[
V\bigl(E_\ell^\star(\alpha),\theta_h\bigr)\leq L\alpha.
\]
Thus, type \(h\)'s utility from purchasing the low-type experiment is nonpositive, and \(r_h(\alpha)=0\). Hence, \(\bar \alpha\) is the threshold at which the low-type experiment becomes an attractive deviation and begins to generate rent for type \(h\). On \([0,\bar \alpha]\), increasing \(\alpha\) raises type \(\ell\)'s payment without requiring any additional rent for type \(h\). Consequently, \eqref{eq:two-type-reduced-profit} is strictly increasing on \([0,\bar \alpha]\), and the seller never chooses \(\alpha<\bar \alpha\).

\begin{proposition}[Optimal menu with opposite defaults]
\label{prop:opposite-default}
Suppose \eqref{eq:opposite-default} holds, and define \(\bar \alpha\) by
\eqref{eq:zbar}. There is an optimal menu in which type \(h\) receives full information and type \(\ell\) receives \(E(1,\alpha)\). Define
\begin{equation}
  \gamma_{\mathrm{opp}}(\tau)
  =
  \frac{(1-\tau)L}{\theta_h-\tau L}
  =
  \frac{(1-\tau)\theta_\ell}
       {\theta_h-\tau\theta_\ell}.
  \label{eq:gamma-opp}
\end{equation}
Because \(\theta_h>L\) and \(\tau<1\),
\(\gamma_{\mathrm{opp}}(\tau)\in(0,1)\).
If \(\gamma<\gamma_{\mathrm{opp}}(\tau)\), the seller sets \(\alpha=1\):
both types receive full information and pay \(L\). If
\(\gamma>\gamma_{\mathrm{opp}}(\tau)\), the seller sets
\(\alpha=\bar \alpha\), with
\[
  t_h=H,
  \qquad
  t_\ell=L\bar \alpha.
\]
Every \(\alpha\in[\bar \alpha,1]\) is optimal at equality. For any such \(\alpha\), the corresponding transfers are
\[
  t_\ell(\alpha)=L\alpha,
  \qquad
  t_h(\alpha)=\theta_h-(\theta_h-L)\alpha.
\]
\end{proposition}

The separating product \(E(1,\bar \alpha)\) is determined entirely by the incentive constraints and does not depend on \(\tau\), with competition only affecting which endpoint the seller prefers. In particular,
\eqref{eq:gamma-opp} shows that the cutoff decreases with
\(\tau\), so stronger
competition expands the separation region.

\paragraph{Comparison of the two configurations}
Type \(h\) receives full information in both configurations. If the
types have the same default, every unit of information supplied to type
\(\ell\) also creates rent for type \(h\), so the seller chooses between full-information pooling and exclusion. If the types have opposite defaults, the seller can increase \(\alpha\) up to \(\bar \alpha\) without creating rent for the high type, so it chooses between partial-information
separation and full-information pooling.

A larger \(\gamma\) places greater weight on preserving the price paid
by type \(h\), while a larger \(\tau\) lowers the seller's return from
improving type \(\ell\)'s accuracy. Both changes favor exclusion or
separation. Because \(\tau<1\), however, an additional unit of buyer
accuracy still raises the firms' combined payoff by \(1-\tau>0\).
Exclusion and partial disclosure therefore arise from screening and
rent extraction, not from negative gains from information.
Figure~\ref{fig:two-type-screening} compares the two configurations at
the same full-information values $H$ and $L$, showing how the required
rent profiles in panel~(a) determine the pooling, exclusion, and
separation regions in panels~(b) and~(c).

\begin{figure}[t]
    \centering
    \includegraphics[width=\textwidth]{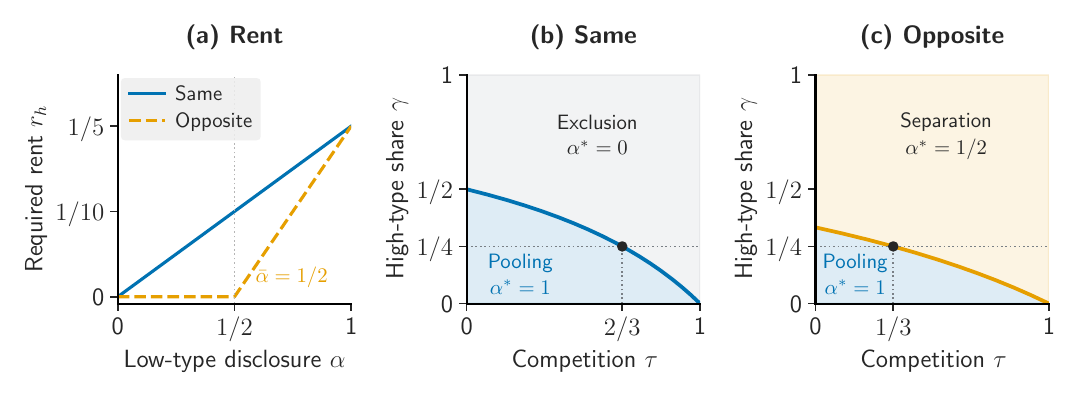}
    \caption{Default actions and optimal screening.
    Both configurations have $\theta_h=0.6$, $H=0.4$, and $L=0.2$;
    $\theta_\ell=0.8$ for the same default and $\theta_\ell=0.2$
    for opposite defaults. (a) Required high-type rent plotted as a function of low-type disclosure $\alpha$.
    (b) With the same default, pooling ($\alpha^*=1$) gives way to
    exclusion ($\alpha^*=0$).
    (c) With opposite defaults, pooling gives way to separation
    ($\alpha^*=1/2$). At $\gamma=1/4$, the respective switches occur
    at $\tau=2/3$ and $\tau=1/3$.
    On the boundaries, all $\alpha\in[0,1]$ in (b), or
    $\alpha\in[1/2,1]$ in (c), are optimal.}
    \label{fig:two-type-screening}
\end{figure}

\section{Continuum of types}
\label{sec:continuum}

We now let \(F\) have full support on \([0,1]\) and a continuous,
strictly positive density \(f\). We maintain the common prior and study
\(0\leq\tau\leq1\); the nontrivial selling region is \(\tau<1\), while
\(\tau=1\) connects the characterization back to optimal withholding.
For \(0\leq\tau<1\), regularity implies that an optimal menu sells
full information at a single price to buyers with sufficiently
uncertain beliefs. When \(F\) is regular at both competitive-exposure
levels being compared, with both below the withholding threshold,
greater exposure raises the price and contracts the interval of
purchasing types.

\paragraph{Implementation}

Propositions~\ref{prop:responsive} and~\ref{prop:concentration} imply that, for
\(0\leq\tau<1\), an optimal menu can be represented by responsive one-sided
experiments. 
When \(\tau=1\), Proposition~\ref{prop:no-trade} shows that withholding is optimal, and admits a responsive one-sided representation. The formulation below therefore also covers this boundary case. 

The following implementability result applies to responsive one-sided menus and is independent of \(\tau\).
Let
\(I=q_1-q_0\in[-1,1]\) denote the experiment's one-sided \textit{signed-distortion} from full information:
%Let \(I=q_1-q_0\in[-1,1]\) record the direction and degree of the one-sided distortion from full information:
\begin{equation*}
  (q_0(I),q_1(I))
  =
  \begin{cases}
    (1,1+I), & I\in[-1,0],\\
    (1-I,1), & I\in[0,1].
  \end{cases}
\end{equation*}

Let \(E_I=E(q_0(I),q_1(I))\), and write
\(V(I,\theta)=V(E_I,\theta)\) and
\(\Delta(I,\theta)=\Delta(E_I,\theta)\).
The sign of \(I\) indicates which state is distorted and \(|I|\) records the amount of distortion: \(I=0\) is full information, whereas both \(I=-1\) and \(I=1\) are uninformative experiments that always recommend actions 0 and 1, respectively.
The upper and right edges in Figure~\ref{fig:experiments}(a) give a
geometric representation of this coordinate: along either edge, moving
toward $E(1,1)$ reduces $|I|$ to zero.
The buyer's incremental value from obeying the recommendations is
\begin{equation}
  \Delta(I,\theta)
  =
  1-\bar u(\theta)
  +
  I\bigl(\theta-\ind\{I\geq0\}\bigr).
  \label{eq:I-value}
\end{equation}

A responsive one-sided direct menu allocates a signed-distortion
\(I(\theta)\in[-1,1]\) and a transfer \(t(\theta)\geq0\) to every type. We say that \(I\) is \emph{responsively implementable} if some nonnegative transfer function makes the direct menu incentive compatible and individually rational, and makes every type weakly willing to obey the recommendations from \(E_{I(\theta)}\).

Define
\begin{equation*}
  \Iset
  =
  \left\{
    I:[0,1]\to[-1,1]:
    I\text{ is nondecreasing and }
    \int_0^1 I(\theta)\,d\theta=0
  \right\}.
\end{equation*}
Because every allocation in \(\Iset\) is monotone and bounded, all of the
integrals below are well defined.

\begin{proposition}[Responsive implementability]
\label{prop:implementability}
A signed-distortion allocation is responsively implementable if and only if it belongs to \(\Iset\). Every \(I\in\Iset\) is responsively
implemented by
\begin{equation}
  t(\theta)
  =
  I(\theta)
  \bigl(\theta-\ind\{I(\theta)\geq0\}\bigr)
  -
  \int_0^\theta I(z)\,dz.
  \label{eq:continuum-transfer}
\end{equation}
The monotonicity and balance conditions ensure incentive
compatibility, individual rationality, obedience, and
\(t(\theta)\geq0\) for every type.
\end{proposition}

The two conditions have distinct economic roles. First, monotonicity is the usual screening restriction: a buyer with a higher posterior belief must receive a weakly greater signed-distortion. Second, the balance condition links the two sides of \(\theta=1/2\). Both types \(\theta=0\) and \(\theta=1\)
cannot benefit from information, so individual rationality and
nonnegative prices require both endpoint types to receive zero
utility. The envelope identity then gives
\begin{equation*}
  U(1)-U(0)
  =
  \int_0^1 I(\theta)\,d\theta
  =
  0.
\end{equation*}
The seller therefore cannot design the menu of experiments offered to low and high posterior beliefs independently.

\paragraph{Virtual surplus}
Integrating the transfer formula in
\eqref{eq:continuum-transfer}, expected incremental profit can be
written, up to the
allocation-independent constant
\(-\tau\E_F[1-\bar u(\theta)]\), as
\begin{equation}
  \int_0^1 \mathcal{J}(I(\theta),\theta)\,d\theta,
  \qquad
  \mathcal{J}(I,\theta)
  =
  \begin{cases}
    I\Phi_-(\theta), & I\leq0,\\
    I\Phi_+(\theta), & I\geq0,
  \end{cases}
  \label{eq:virtual-objective}
\end{equation}
where
\begin{align}
  \Phi_-(\theta)
  &=
  F(\theta)+(1-\tau)\theta f(\theta),
  \label{eq:phi-minus}\\
  \Phi_+(\theta)
  &=
  F(\theta)-(1-\tau)(1-\theta)f(\theta).
  \label{eq:phi-plus}
\end{align}

The integral in \eqref{eq:virtual-objective} is just an integral over posterior beliefs. The density \(f\) enters through
\(\Phi_-\) and \(\Phi_+\), which measure the seller's marginal return
from changing information on the two sides of full information.

At \(I=0\), the left and right slopes are respectively
\(\Phi_-(\theta)\) and \(\Phi_+(\theta)\), and their difference is
\begin{equation}
  \Phi_-(\theta)-\Phi_+(\theta)
  =
  (1-\tau)f(\theta).
  \label{eq:phi-gap}
\end{equation}
Because \(f(\theta)>0\), the pointwise objective
\(\mathcal{J}(I,\theta)\) is concave in the signed-distortion allocation \(I\) if and
only if \(\tau\leq1\). It has a strict downward kink at \(I=0\) when
\(\tau<1\) and is linear in \(I\) when \(\tau=1\). 

For \(\tau>1\), the kink is upward so the reduced objective is convex, and the concave-ironing method used
below does not apply. However, there is no gap in the economic
characterization because Proposition~\ref{prop:no-trade} already
establishes that withholding is optimal in that region.

The seller's problem is consequently equivalent, up to the constant
above, to
\begin{equation}
  \max_{I\in\Iset}
  \int_0^1 \mathcal{J}(I(\theta),\theta)\,d\theta.
  \label{eq:reduced-continuum-problem}
\end{equation}

\paragraph{Regularity}

The posterior distribution \(F\) is \emph{regular at \(\tau\)} if
\(\Phi_-\) and \(\Phi_+\) are continuous and strictly increasing on
\([0,1]\).

If \(f\) is continuously differentiable, a sufficient condition for
regularity is that both
\begin{align}
  \Phi_-'(\theta)
  &=
  (2-\tau)f(\theta)
  +(1-\tau)\theta f'(\theta),
  \label{eq:phi-minus-derivative}\\
  \Phi_+'(\theta)
  &=
  (2-\tau)f(\theta)
  -(1-\tau)(1-\theta)f'(\theta)
  \label{eq:phi-plus-derivative}
\end{align}
are strictly positive.

\begin{proposition}[Optimal continuum menu under regularity]
\label{prop:regular-continuum}
Suppose \(F\) is regular at \(\tau\) and \(0\leq\tau<1\). There is a
unique cutoff \(c\in(0,1/2)\) satisfying
\begin{equation}
  \Phi_-(c)=\Phi_+(1-c).
  \label{eq:cutoff}
\end{equation}
The associated Lagrange multiplier on the integral constraint satisfies
\begin{equation*}
  \lambda^{\star}
  =
  \Phi_-(c)
  =
  \Phi_+(1-c).
\end{equation*}
Equivalently,
\begin{equation*}
  c=\Phi_-^{-1}(\lambda^{\star}),
  \qquad
  1-c=\Phi_+^{-1}(\lambda^{\star}).
\end{equation*}
Thus \(c\) is the allocation cutoff (and, below, the posted price) whereas
\(\lambda^{\star}\) is the shadow value that clears the balance condition (integral constraint).
An optimal signed-distortion allocation is
\begin{equation}
  I^{\star}(\theta)
  =
  \begin{cases}
    -1, & \theta<c,\\
    0, & c\leq\theta\leq1-c,\\
    1, & \theta>1-c.
  \end{cases}
  \label{eq:regular-allocation}
\end{equation}
It is implemented by selling full information at price \(c\) to types
\(\theta\in[c,1-c]\) and allocating the appropriate uninformative
experiment outside this interval.

Within the responsive one-sided representation, \eqref{eq:regular-allocation} and its induced transfer are unique \(F\)-almost everywhere; the choices at the two cutoff types are payoff-irrelevant.

%Except for the products allocated to the two cutoff types, the optimal allocation and corresponding prices are unique. Alternative zero-rent products may be allocated at the cutoffs without changing expected profit. No partially informative experiment is allocated to a positive-probability set of types.
\end{proposition}

For a generic regular distribution, Proposition~\ref{prop:regular-continuum} provides an exact but generally implicit solution: the cutoff $c$ is the unique root of \eqref{eq:cutoff} and can therefore be computed by a one-dimensional root search. An elementary closed-form expression is available only for special parametric distributions.

The seller therefore uses a single posted price for full information.
Buyers with intermediate beliefs purchase because they have the
greatest residual uncertainty. Buyers in either tail are sufficiently
confident that serving them is not worth the resulting competitive
cost and information-rent consequences.

The two excluded intervals have the same length because the balance
condition requires
\[
  \int_0^1 I^{\star}(\theta)\,d\theta=-c+c=0.
\]
Their probability masses need not be equal when \(F\) is asymmetric.
Figure~\ref{fig:posted-price} illustrates this implementation at
$c=1/3$, relating the balanced signed-distortion allocation to the
posted price and the resulting information rents.
\begin{figure}[t]
    \centering
    \includegraphics[width=\textwidth]{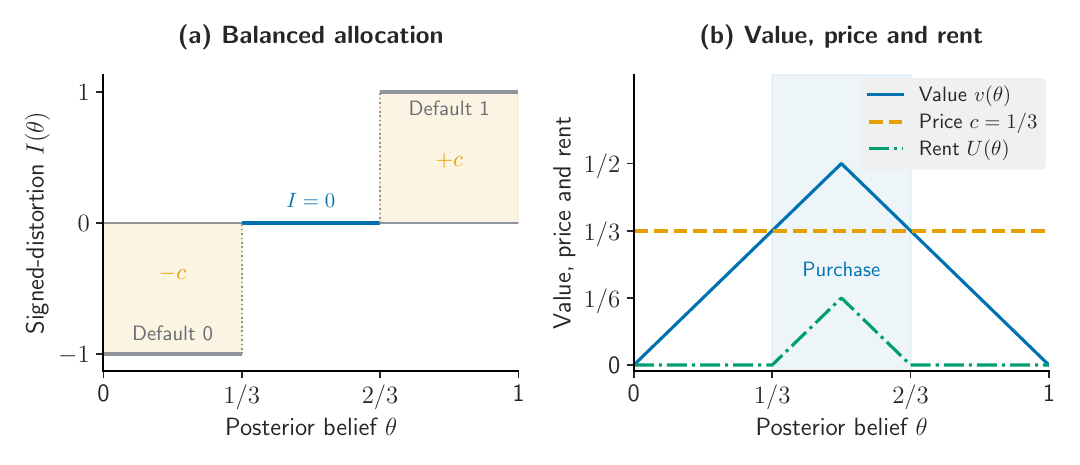}
    \caption{Posted-price implementation and balance at $c=1/3$.
    (a) Full information is allocated on $[1/3,2/3]$, with
    uninformative products outside. The shaded signed areas cancel,
    so $\int_0^1 I(\theta)\,d\theta=0$; balance concerns posterior-space
    lengths rather than probability masses under $F$.
    (b) Full-information value $v(\theta)=\min\{\theta,1-\theta\}$
    (solid), the offered price $c$ (dashed), and buyer rent
    $U(\theta)=\max\{v(\theta)-c,0\}$ (dash-dot).
    The blue shading marks the purchase interval.
    This price is optimal for uniform types at $\tau=1/2$.}
    \label{fig:posted-price}
\end{figure}

If the density is symmetric around one half, the cutoff equation
simplifies to
\begin{equation*}
  F(c)+(1-\tau)cf(c)=\frac12.
\end{equation*}
Consequently, \(\lambda^{\star}=1/2\) for every symmetric regular
distribution.

\begin{corollary}[Competition contracts the served interval]
\label{cor:cutoff-response-to-competition}
Suppose \(F\) is regular at \(\tau_1\) and \(\tau_2\), where
\(0\leq\tau_1<\tau_2<1\), and let \(c(\tau)\) denote the cutoff in
Proposition~\ref{prop:regular-continuum}. Then
\begin{equation*}
  c(\tau_2)>c(\tau_1).
\end{equation*}
Thus stronger competition strictly contracts the interval of buyers
purchasing full information whenever regularity is maintained.
\end{corollary}

\begin{proof}
The equality in \eqref{eq:cutoff} is equivalent to
\begin{equation*}
  F(c)-F(1-c)
  +(1-\tau)c\bigl[f(c)+f(1-c)\bigr]
  =
  0.
\end{equation*}
Regularity makes the left-hand side strictly increasing in \(c\). Increasing
\(\tau\) lowers the seller's net return from buyer accuracy, so the
equality is restored at a higher cutoff.
\end{proof}

\begin{corollary}[Uniform posterior types]
\label{cor:uniform}
Let \(F(\theta)=\theta\). For \(0\leq\tau<1\), the seller charges
\begin{equation}
  t^{\star}(\tau)
  =
  \frac{1}{2(2-\tau)}
  \label{eq:uniform-price}
\end{equation}
for full information and sells it to
\begin{equation}
  \theta
  \in
  [t^{\star}(\tau),1-t^{\star}(\tau)].
  \label{eq:uniform-interval}
\end{equation}
The seller allocates an uninformative product outside this interval.
Expected incremental profit is
\begin{equation}
  \Pi^{\star}(\tau)
  =
  \frac{(1-\tau)^2}{4(2-\tau)}.
  \label{eq:uniform-profit}
\end{equation}
Apart from the products allocated at the cutoff types, this menu is
unique.
\end{corollary}

For the uniform distribution, $\lambda^{\star}=1/2$, so
$c(\tau)=t^{\star}(\tau)=1/[2(2-\tau)]$.
Figure~\ref{fig:uniform-competition} shows the resulting comparative
statics. Stronger competition raises the posted price and contracts
the set of purchasers around $\theta=1/2$, while the information sold
to each purchaser remains fully revealing. As $\tau\uparrow1$, the
purchase probability and incremental profit both converge to zero.
At $\tau=1$, withholding is optimal by Proposition~\ref{prop:no-trade}.

\begin{figure}[t]
    \centering
    \includegraphics[width=\textwidth]{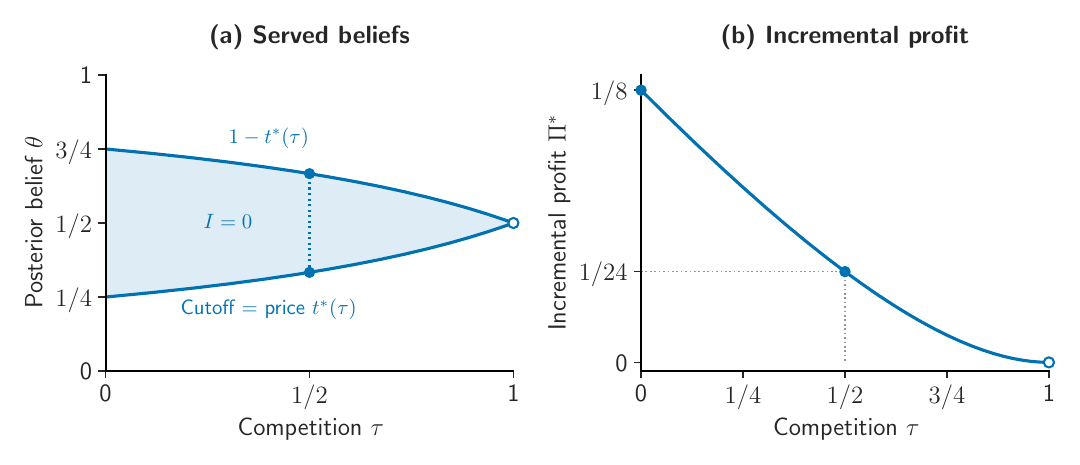}
    \caption{Competition and optimal information sales with uniform types.
    (a) The curves $c(\tau)=t^{\star}(\tau)=1/[2(2-\tau)]$ and
    $1-c(\tau)$ bound the interval receiving full information
    ($I^{\star}=0$); the outer regions receive uninformative products.
    (b) Expected incremental profit relative to withholding,
    $\Pi^{\star}(\tau)=(1-\tau)^2/[4(2-\tau)]$.
    At $\tau=1/2$, the served interval is $[1/3,2/3]$ and profit is
    $1/24$. Open endpoints indicate limits as $\tau\uparrow1$.
    At $\tau=1$, withholding is optimal.}
    \label{fig:uniform-competition}
\end{figure}

\paragraph{Beyond regularity}
Without regularity, either virtual value may decrease with the buyer's posterior belief. The allocation suggested by pointwise virtual-surplus maximization can then be nonmonotone and so not responsively implementable by Proposition~\ref{prop:implementability}.
The concavity established above nevertheless permits an exact solution for every \(0\leq\tau\leq1\).

The solution has three steps. First, the balance condition is relaxed with a scalar multiplier. Second, for a fixed multiplier, we iron the two linear branches of the virtual objective to impose monotonicity. Finally, a multiplier is chosen to restore balance. The regular solution is the special case in which neither branch requires ironing.
Let
\[
  \Iset^\uparrow
  =
  \left\{
    I:[0,1]\to[-1,1]:
    I\text{ is nondecreasing}
  \right\}
\]
denote the set of monotone allocations before imposing balance. Write
\[
  B(I)=\int_0^1 I(\theta)\,d\theta.
\]
The balance condition is \(B(I)=0\), so the implementable set is
\[
  \Iset
  =
  \left\{
    I\in\Iset^\uparrow:B(I)=0
  \right\}.
\]
For \(\lambda\in\mathbb R\), relaxing balance gives the Lagrangian
\begin{equation*}
  \mathcal L_\lambda(I)
  =
  \int_0^1
  \left[
    \mathcal{J}(I(\theta),\theta)-\lambda I(\theta)
  \right]\,d\theta.
\end{equation*}
Because \(-\lambda I\) is linear in \(I\), the concavity established above for \(0\leq\tau\leq1\) is maintained.

\begin{proposition}[Dual characterization under irregularity]
\label{prop:irregular}
Suppose \(F\) has a continuous, strictly positive density and
\(0\leq\tau\leq1\). Problem
\eqref{eq:reduced-continuum-problem} admits an optimal allocation. Moreover, a balanced allocation \(I^{\star}\in\Iset\) is optimal if and only if
there exists an unrestricted multiplier \(\lambda^{\star}\in\mathbb R\)
such that
\begin{equation*}
  I^{\star}
  \in
  \arg\max_{I\in\Iset^\uparrow}
  \mathcal L_{\lambda^{\star}}(I).
\end{equation*}
At least one such primal-dual pair \((I^{\star},\lambda^{\star})\) exists.
\end{proposition}

Because balance is an equality constraint, its multiplier may take any
real value. Proposition~\ref{prop:irregular} separates the problem into
two parts. For a fixed \(\lambda\), the seller maximizes
\(\mathcal L_\lambda\) subject only to monotonicity. The multiplier is then adjusted until the allocation is balanced. The proposition establishes that every optimum has such a supporting multiplier.
%, rather than just providing a sufficient optimality test.

%\paragraph{Branchwise ironing}
For fixed \(\lambda\), the shifted integrand is linear on each side of
full information. Its slope is \(\Phi_-(\theta)-\lambda\) for \(I<0\) and \(\Phi_+(\theta)-\lambda\) for \(I>0\). If either virtual value is nonmonotone, its pointwise maximizer may violate the monotonicity of
\(I\). We thus follow the convex-minorant method of
\citet{toikka2011ironing} and iron the two branches separately. This branchwise construction accommodates the kink at \(I=0\); the appendix provides a direct proof.

Define the two cumulative virtual values by
\[
  \mathcal H_\pm(y)
  :=
  \int_0^y\Phi_\pm(\theta)\,d\theta,
\]
with \(y\in[0,1]\). For each sign, let \(\mathcal C_\pm\) be the greatest convex function
lying weakly below \(\mathcal H_\pm\), and let
\(\widetilde\Phi_\pm\) be its right derivative, using the left derivative at \(y=1\). Convexity of \(\mathcal C_\pm\) makes \(\widetilde\Phi_\pm\) nondecreasing. When \(\mathcal C_\pm\) replaces
a portion of \(\mathcal H_\pm\) by a straight line, its slope \(\widetilde\Phi_\pm\) is constant on that interval. The resulting flat ironed value can lead a range of distinct types to receive the same information product.

Ironing preserves the ordering and any uniform gap between the two virtual values. Consequently, \eqref{eq:phi-gap} implies
\[
  \widetilde\Phi_-(\theta)
  \geq
  \widetilde\Phi_+(\theta).
\]

Therefore, \(I=-1\) when \(\widetilde\Phi_-\) is below \(\lambda\), \(I=0\) when \(\lambda\) lies between the two ironed virtual values, and \(I=1\)
when \(\widetilde\Phi_+\) is above \(\lambda\). The following allocations resolve equalities toward, respectively, the smallest and largest optimal $I$.
For each \(\lambda\in\mathbb R\), define
\[
\underline I_\lambda(\theta)
=
\begin{cases}
  -1,
  & \widetilde\Phi_-(\theta)\leq\lambda,\\
  0,
  & \widetilde\Phi_-(\theta)>\lambda
    \geq\widetilde\Phi_+(\theta),\\
  1,
  & \widetilde\Phi_+(\theta)>\lambda,
\end{cases}
\]
and
\[
\overline I_\lambda(\theta)
=
\begin{cases}
  -1,
  & \widetilde\Phi_-(\theta)<\lambda,\\
  0,
  & \widetilde\Phi_-(\theta)\geq\lambda
    >\widetilde\Phi_+(\theta),\\
  1,
  & \widetilde\Phi_+(\theta)\geq\lambda.
\end{cases}
\]
These allocations differ only where an ironed virtual value equals
\(\lambda\).

\begin{proposition}[Constructive solution under irregularity]
\label{prop:ironing}
Suppose \(F\) has a continuous, strictly positive density and
\(0\leq\tau\leq1\). For every \(\lambda\in\mathbb R\), the allocations
\(\underline I_\lambda\) and \(\overline I_\lambda\) both maximize
\(\mathcal L_\lambda\) over \(\Iset^\uparrow\). Any other maximizing
allocation lies between them, except possibly on a zero-probability set
of types.
There exists a multiplier \(\lambda^{\star}\in\mathbb R\) such that
\[
  B(\underline I_{\lambda^{\star}})
  \leq0
  \leq
  B(\overline I_{\lambda^{\star}}).
\]
If these two balance residuals differ, define
\[
  \beta
  =
  \frac{-B(\underline I_{\lambda^{\star}})}
       {B(\overline I_{\lambda^{\star}})
        -B(\underline I_{\lambda^{\star}})}.
\]
Then
\[
  I^{\star}
  =
  (1-\beta)\underline I_{\lambda^{\star}}
  +
  \beta\overline I_{\lambda^{\star}}
\]
is balanced and solves
\eqref{eq:reduced-continuum-problem}. If the two balance residuals
coincide, both extremal allocations are already balanced and optimal.
\end{proposition}

Proposition~\ref{prop:ironing} yields a direct algorithm. Compute the two convex minorants once, construct the extremal allocations for a candidate multiplier, and vary \(\lambda\) until their balance residuals bracket zero. Both residuals are weakly decreasing in \(\lambda\), so this is a one-dimensional search. If an extremal allocation is balanced, it is optimal. Otherwise, the convex combination in the proposition restores balance.

If zero lies strictly between the two balance residuals, the balancing coefficient satisfies \(\beta\in(0,1)\). If \(\tau<1\), the strict gap between the two ironed branches implies that, wherever the extremal allocations differ, \(\left(\underline I_{\lambda^\star}(\theta),
\overline I_{\lambda^\star}(\theta)\right) \in \left\{(-1, 0), (0, 1)\right\}\). At every such type, the balanced convex combination \((1-\beta)\underline I_{\lambda^\star}(\theta) +
\beta\overline I_{\lambda^\star}(\theta)\) therefore lies in \((-1,0)\) or \((0,1)\), respectively, and hence
corresponds to a partially informative experiment. 
At \(\tau=1\), both ironed virtual
values equal the strictly increasing function \(F\), so the extremal allocations can differ only on a zero-measure level set. Under regularity, no ironing occurs:
\(\mathcal C_\pm=\mathcal H_\pm\) and
\(\widetilde\Phi_\pm=\Phi_\pm\). The algorithm then reduces
to the three-valued cutoff allocation in Proposition~\ref{prop:regular-continuum}.

At \(\tau=0\), the reduced problem and its ironing solution coincide
with the benchmark of \citet{bergemann2018design}, which applies the
convex-minorant method of \citet{toikka2011ironing}. Our extension
replaces the benchmark virtual values with the competition-adjusted
values \(\Phi_-\) and \(\Phi_+\) and establishes that the branchwise
construction remains valid for every \(0\leq\tau\leq1\), while
accommodating the kink at full information when \(\tau<1\).

At the boundary \(\tau=1\), \(\Phi_-=\Phi_+=F\). Taking
\(\lambda^{\star}=F(1/2)\) gives
\[
  I^{\star}(\theta)
  =
  \begin{cases}
    -1, & \theta<1/2,\\
    1, & \theta>1/2,
  \end{cases}
\]
with the value at \(\theta=1/2\) immaterial. This allocation is balanced
and represents withholding, so no ironing is required at the no-trade
threshold.

\paragraph{Regularization by competition}
The preceding construction solves the problem when the virtual values are irregular, but competition itself can also eliminate the need for ironing.
In \eqref{eq:phi-minus-derivative} and
\eqref{eq:phi-plus-derivative}, the terms involving \(f'\), which can
break monotonicity, are multiplied by \(1-\tau\). They vanish as
\(\tau\) approaches one, while the positive terms involving \(f\)
remain. The following result formalizes this observation.

\begin{corollary}[Competition eventually restores regularity]
\label{cor:eventual-regularity}
Suppose \(f\) is continuously differentiable and strictly positive on
\([0,1]\), and define
\[
  f_{\min}
  :=
  \min_{\theta\in[0,1]}f(\theta)>0,
  \qquad
  \lVert f'\rVert_\infty
  :=
  \max_{\theta\in[0,1]}|f'(\theta)|.
\]
For \(0\leq\tau\leq1\), if
\begin{equation}
  (2-\tau)f_{\min}
  >
  (1-\tau)\lVert f'\rVert_\infty,
  \label{eq:eventual-regularity-condition}
\end{equation}
then \(F\) is regular at \(\tau\). Consequently, for every continuously
differentiable density that is strictly positive on \([0,1]\), there
exists \(\bar\tau<1\) such that \(F\) is regular for every
\(\tau\in(\bar\tau,1]\).
\end{corollary}

Indeed, \eqref{eq:phi-minus-derivative} and
\eqref{eq:phi-plus-derivative} imply, for both signs,
\[
  \Phi_\pm'(\theta)
  \geq
  (2-\tau)f_{\min}
  -(1-\tau)\lVert f'\rVert_\infty,
\]
such that the condition in \eqref{eq:eventual-regularity-condition} therefore makes both virtual values strictly increasing. The lower bound converges to \(f_{\min}>0\) as \(\tau\uparrow1\), which proves the eventual-regularity
claim.

This result identifies an effect absent from the no-competition benchmark: intense competition can eliminate the ironing generated by an irregular type distribution. For every \(\tau\in(\bar\tau,1)\), the general ironing solution reduces to the single posted-price menu in Proposition~\ref{prop:regular-continuum}; at \(\tau=1\), withholding is optimal.

\paragraph{Example}
Let
\[
  f(\theta)=1+\frac35\sin(4\pi\theta).
\]
This smooth density integrates to one and is strictly positive. At
\(\tau=0\), it is irregular as
\[
  \Phi_+'(1/2)=2-\frac{6\pi}{5}<0.
\]

A numerical application of Proposition~\ref{prop:ironing} yields an ironing interval approximately equal to \(\theta\in[0.22,0.38]\), with the types in this interval all receiving the same partially informative experiment
\(E(1,0.78)\).

At \(\tau=0.95\), by contrast,
\(f_{\min}=2/5\) and
\(\lVert f'\rVert_\infty=12\pi/5\), so \eqref{eq:eventual-regularity-condition} holds. The
distribution is therefore regular, and the cutoff equation gives
\[
  c=\frac{10}{21}.
\]
Figure~\ref{fig:competition-regularization} compares the two optimal
allocations. At $\tau=0$, the active ironed interval supports a common
partially informative product. At $\tau=0.95$, full information is
sold at price $10/21$ on $[10/21,11/21]$, with no useful information
outside this interval. In this example, stronger competition eliminates
the partially informative product while contracting the set of types
receiving useful information.

\begin{figure}[t]
    \centering
    \includegraphics[width=\textwidth]{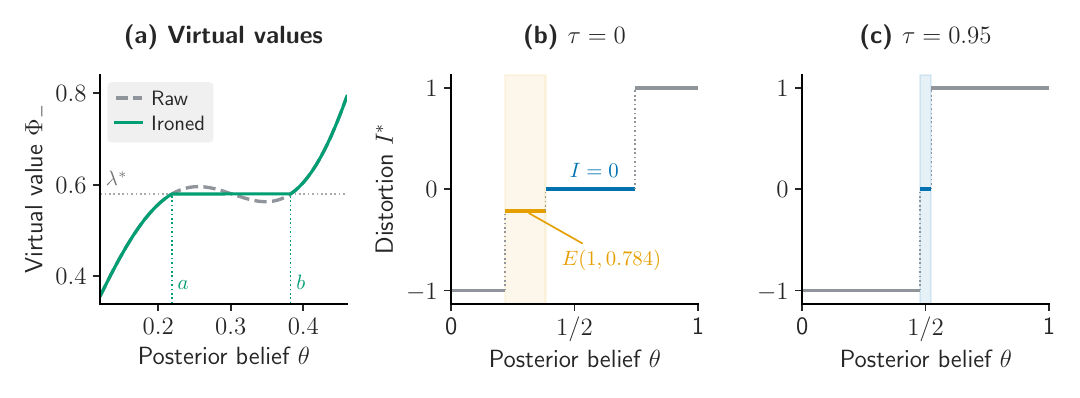}
    \caption{Ironing and regularization for
    $f(\theta)=1+0.6\sin(4\pi\theta)$.
    (a) At $\tau=0$, the raw virtual value $\Phi_-$ (dashed) and its
    ironed counterpart (solid), shown on the relevant interval.
    The active plateau spans $[a,b]\simeq[0.219,0.382]$ at
    $\lambda^{\star}\simeq0.580$.
    (b) The corresponding optimal allocation assigns
    $E(1,0.784)$ on this interval, with $I\simeq-0.216$, and full
    information from $b$ to $d\simeq0.746$.
    (c) At $\tau=0.95$, both virtual values are strictly increasing:
    the optimal allocation supplies full information only on
    $[10/21,11/21]$, at price $10/21$.
    The outer intervals in both allocations receive uninformative products.}
    \label{fig:competition-regularization}
\end{figure}

\section{Conclusions}
\label{sec:conclusion}
Selling information to a competitor can remain profitable even when
better buyer decisions harm the seller. Under a common prior, an
increase of \(\varepsilon\) in the buyer's probability of choosing
correctly creates buyer value \(\varepsilon\) and seller cost
\(\tau\varepsilon\), where \(\tau\) measures competitive exposure.
The incremental joint surplus is therefore
\((1-\tau)\varepsilon\), while the seller's incremental profit at
price \(t\) is \(t-\tau\varepsilon\). Withholding is optimal when
competitive exposure reaches or exceeds the buyer's unit benefit
from choosing correctly.

Holding the distribution of buyer beliefs fixed, greater competitive
exposure weakly reduces both the seller's optimal profit relative
to withholding and the buyer's overall probability of choosing
correctly. These conclusions do not require regularity of the
belief distribution. They concern average outcomes across buyer
types and do not imply that every individual information product
becomes less informative.

The detailed characterizations explain how the optimal menu changes.
In the two-type model, the more uncertain type receives full
information below the withholding threshold. Greater competitive
exposure favors withholding from the more certain type when the
types would otherwise choose the same action, and offering it a
distinct product when they would choose different actions. In the
latter case, an optimal partial experiment can be chosen independently
of competitive exposure, although other experiments may also be
optimal where policies tie.

For a continuum of beliefs, regularity yields full information at a
single price below the withholding threshold. When regularity holds
at both exposure levels compared, greater exposure raises the price
and reduces the set of purchasers. We also characterize optimal menus
when regularity fails. For every continuously differentiable density
that is strictly positive throughout the closed unit interval, full
information at a single price becomes optimal sufficiently close to
the withholding threshold from below.

These results extend the information-screening framework of
\citet{bergemann2018design} to a seller that bears a direct cost from
improving the buyer's decisions. The distinction from information
sales to competing buyers, as in \citet{bonatti2024selling}, is that
the seller itself incurs the competitive loss.

Our characterizations rely on a simple environment: common prior beliefs, binary states and actions, a one-dimensional type, a single buyer, finite-support experiments, perfect seller information about the realized state, nonnegative outcome-independent prices, risk neutrality, a known type distribution, and \textit{ex ante} commitment. 

The common-prior assumption is especially important, ensuring the buyer's incremental value and the seller's competitive loss are two evaluations of the same increase in decision accuracy. If they instead hold different conditional beliefs, the buyer values an experiment under one distribution while the seller evaluates its competitive consequences under another. Willingness to pay may rise without a proportional increase in seller-assessed loss. The scalar margin \(1-\tau\) no longer summarizes the effect of additional accuracy, thus neither the withholding threshold in Proposition~\ref{prop:no-trade} nor the one-sided improvement argument in
Proposition~\ref{prop:concentration} extends mechanically. The mechanism must instead account separately for the parties' valuations of each experiment.

A nonseparable downstream game breaks the same reduction for a different reason. If disclosure affects prices, quantities, entry, or other strategic choices, the seller's loss depends on the particular equilibrium induced by the buyer's posterior rather than only on the probability of a correct action. Buyer value and seller cost may then be nonlinear, nonmonotone, or discontinuous in the experiment. The mechanism design problem must be solved jointly with the downstream Bayesian game, including equilibrium selection when the continuation game has multiple equilibria.

The problem becomes richer still with several buyers. Information allocated to one firm changes the strategic environment and may alter the value of information allocated to its competitors. One buyer's report may therefore affect both the signals received by other buyers and their continuation payoffs. The seller must design a joint information structure while accounting for informational exclusivity, cross-buyer deviations, and joint incentive and obedience constraints. The bilateral screening problem then becomes a multidimensional mechanism design problem in which information and strategic externalities are jointly determined.

Binary states and actions provide another major source of
tractability. With a binary state, the buyer's posterior is summarized by the scalar \(\theta\), which gives types a natural order. With three or more states, posterior beliefs lie in a higher-dimensional simplex and experiments are higher-dimensional stochastic matrices. Scalar monotonicity is generally replaced by multidimensional incentive constraints and cyclic-monotonicity conditions \citep{rochet1987necessary}, which control deviations around finite cycles of types rather than only deviations between ordered types.

Allowing the buyer to choose amongst three or more actions makes incentive compatibility harder to verify. To check incentive compatibility, the seller must ask what a buyer would gain by selecting a menu item intended for another type. In the binary case, a buyer who selects such an item has only three relevant ways to use its signals: follow the recommendations, reverse them, or ignore them and take its default action. For a responsive binary experiment, reversing the recommendations never outperforms the default action. It is therefore enough to compare following the recommendations with ignoring
them. With three or more available actions, however, the buyer may choose an unintended action after each signal. The seller must then consider both which menu item the buyer selects and how it would respond to every signal generated by that item. Consequently, the binary deviation formula, the one-sided reduction, and the scalar signed-distortion allocation need not extend.

Our analysis restricts the buyer to purchasing at most one experiment from the menu. Allowing the buyer to combine menu items would require specifying the joint distribution of their signals conditional on the state and accounting for deviations involving multiple purchases. In our setting, these deviations can change both total payments and the competitive cost of the buyer’s resulting gain in decision accuracy. Characterizing optimal menus when multiple purchases are permitted, and identifying when the menus derived here remain incentive compatible, are natural directions for future work.

Other extensions change the feasible information products. If the seller observes only a noisy raw signal rather than the state itself, each offered experiment must be obtainable as a Blackwell garbling of that signal. Experiment design is then constrained by the entire joint information structure connecting the state, the seller's data, and the buyer's private signal; the seller can no longer choose arbitrary state-dependent signal distributions.

The timing of information creates further signaling problems. Here, the seller commits before observing the state, so the posted menu does not itself reveal state information. If the seller is already privately informed when posting, the menu can signal that
information and the problem becomes one of informed principal mechanism design. Similarly, if the seller observes the buyer's purchase before taking a strategically relevant action, that choice becomes an endogenous signal of the buyer's information. The buyer's incentive constraints, the seller's behavior, and Bayes-consistent beliefs must then be determined jointly.

%Limited commitment creates a distinct difficulty. After observing the state, the seller may prefer to alter the promised experiment or withhold a signal that has become competitively costly to disclose. A feasible dynamic mechanism must therefore satisfy sequential credibility in addition to participation, incentive compatibility, and obedience. Reputation, repeated interaction, or an external commitment technology may then be needed to support the intended disclosure policy.

%Contingent transfers could enlarge the contracting set, but only when the relevant states, signals, or actions are verifiable. State- or action-contingent contracts, including scoring-rule-like schemes, also require explicit assumptions about budgets, limited liability, and collateral. Dynamic interaction alone does not make such transfers feasible; contractibility and commitment must be specified separately.

Finally, the continuum characterization assumes that the type distribution \(F\) is known. Under distributional uncertainty, the seller would need a mechanism robust across a class of distributions consistent with the common prior, or would need to learn \(F\) through repeated interaction. Learning is complicated because observed
contract choices are endogenous to the mechanism and experimentation may itself impose competitive cost. Additional buyer characteristics, such as heterogeneous timing preferences, information-processing abilities, or values of decision accuracy, would likewise make the
posterior belief only one component of a multidimensional private type. Risk aversion or downside-risk constraints would make the distribution of transfers and competitive losses relevant, invalidating an objective based only on expected incremental profit.

The value of this work is thus not that its exact menu formulas should remain unchanged when selling information to a competitor in richer environments. Rather, it identifies the structures responsible for tractability: the common prior makes demand and competitive loss proportional, binary states and actions make screening one-dimensional, and commitment separates information-product design from subsequent strategic responses. Relaxing one of these structures requires replacing the corresponding argument rather than merely adding another comparative
static. %This work nevertheless yields a robust finding: when supplying decision accuracy also strengthens a competitor, more intense comepition makes the optimal information offering more selective, even though the screening constraints that govern feasible products remain unchanged.

\section*{Acknowledgements}
The authors wish to thank Lesia Mitridati, Rachel Cummings, and Munther Dahleh for fruitful discussions and insightful comments that greatly improved this work.

\appendix
\section{Proofs}
\label{app:proofs}

%\subsection{Mechanism design}

\begin{proof}[\textbf{Proof of Proposition~\ref{prop:responsive}}] 
For every report \(\hat\theta\) and signal \(s\), let
\[
a_{\hat\theta}(s)
=
\begin{cases}
0,
& (1-\hat\theta)\sigma_{\hat\theta}(s\mid0)
  \geq \hat\theta\sigma_{\hat\theta}(s\mid1),\\
1,
& (1-\hat\theta)\sigma_{\hat\theta}(s\mid0)
  < \hat\theta\sigma_{\hat\theta}(s\mid1).
\end{cases}
\]
Pool signals according to this action, while retaining the fixed
common signal space \(\mathcal S\), by defining, for \(a,x\in\{0,1\}\),
\[
\sigma^{\mathrm r}_{\hat\theta}(s_a\mid x)
=
\sum_{s\in\mathcal S}
\ind\{a_{\hat\theta}(s)=a\}\sigma_{\hat\theta}(s\mid x),
\]
and setting
\(\sigma^{\mathrm r}_{\hat\theta}(s\mid x)=0\) for
\(s\notin\{s_0,s_1\}\).
%The likelihoods in the original direct menu are Borel in \(\hat\theta\); hence each indicator above is Borel, and so is each countable nonnegative sum. Thus the pooled allocation remains a Borel family of finite-support experiments.}

The map \(s\mapsto s_{a_{\hat\theta}(s)}\) is a Blackwell garbling.
Let
\(\mathcal{A}_a=\{s:a_{\hat\theta}(s)=a\}\). Because action \(a\) is optimal at
every signal in \(\mathcal{A}_a\), summing the signalwise inequalities gives
\[
\begin{aligned}
&\max\!\left\{
 (1-\hat\theta)\sum_{s\in \mathcal{A}_a}\sigma_{\hat\theta}(s\mid0),
 \hat\theta\sum_{s\in \mathcal{A}_a}\sigma_{\hat\theta}(s\mid1)
\right\}=
\sum_{s\in \mathcal{A}_a}
\max\!\left\{
 (1-\hat\theta)\sigma_{\hat\theta}(s\mid0),
 \hat\theta\sigma_{\hat\theta}(s\mid1)
\right\},
\end{aligned}
\]
where the maximum on the left is attained by action \(a\). Summing
over \(a\) shows that pooling preserves the intended type's sum of
signalwise maxima, value, and matching probability, and it also proves
obedience to each pooled recommendation. Hence the competitive cost is
unchanged. Finally, any response rule that another type \(\theta\)
uses after a pooled recommendation can be applied to the original
signal through \(s\mapsto s_{a_{\hat\theta}(s)}\). Optimization in the
original experiment therefore gives
\[
V(E^{\mathrm r}_{\hat\theta},\theta)
\leq V(E_{\hat\theta},\theta).
\]
Keeping transfers fixed therefore preserves incentive compatibility,
since
\[
\begin{aligned}
V(E^{\mathrm r}_{\theta},\theta)-t(\theta)
&=V(E_{\theta},\theta)-t(\theta)\\
&\geq V(E_{\hat\theta},\theta)-t(\hat\theta)\\
&\geq V(E^{\mathrm r}_{\hat\theta},\theta)-t(\hat\theta).
\end{aligned}
\]
Truthful utilities, nonnegative transfers, and competitive costs are
unchanged, which proves feasibility and profit equivalence.
\end{proof}

\begin{proof}[\textbf{Proof of Lemma~\ref{lem:deviation}}]
There are four deterministic response rules: follow, reverse, and the two constant actions. Following yields \(\Delta(E,\theta)\); the default constant yields zero; and the opposite constant is weakly worse than the default. Let
\(C(E,\theta)=(1-\theta)q_0+\theta q_1\) denote the matching probability from following. Because \(q_0+q_1\geq1\),
\[
C(E,\theta)
\geq \min\{\theta,1-\theta\}(q_0+q_1)
\geq \min\{\theta,1-\theta\}.
\]
Reversing therefore yields matching probability
\(1-C(E,\theta)\leq\bar u(\theta)\), and hence weakly negative incremental value. Randomization cannot improve on the best deterministic rule, so the buyer chooses between following and its default, yielding \eqref{eq:deviation-value}.
\end{proof}

\begin{proof}[\textbf{Proof of Proposition~\ref{prop:no-trade}}]
For every type, individual rationality gives
\[
t(\theta)-\tau V(E_\theta,\theta)
=(1-\tau)V(E_\theta,\theta)-U(\theta),
\]
which is nonpositive when \(\tau\geq1\), whereas
withholding yields zero. If \(\tau>1\) and
\(V(E_\theta,\theta)>0\) on a set of positive \(F\)-measure, then
\[
  \Pi(\mathcal M)
  \leq
  -(\tau-1)\int_\Theta V(E_\theta,\theta)\,dF(\theta)
  <0.
\]
At \(\tau=1\), expected incremental profit equals \(-\int_\Theta U(\theta)\,dF(\theta)\), which is zero if and only if the expected rent is zero.
\end{proof}

\begin{proof}[\textbf{Proof of Proposition~\ref{prop:concentration}}]
At report \(\theta\), let
\(\varepsilon=\min\{1-q_0,1-q_1\}\) and raise \(q_0\), \(q_1\), and the
price by \(\varepsilon\). Each obedience margin rises by
\(\varepsilon\), while the intended type's value and price both rise by
\((1-\theta)\varepsilon+\theta\varepsilon=\varepsilon\); its truthful
utility is therefore unchanged. For any deviating type \(\theta'\),
Lemma~\ref{lem:deviation} implies
\[
V(E',\theta')
=\max\{\Delta(E,\theta')+\varepsilon,0\}
\leq V(E,\theta')+\varepsilon.
\]
Because the price also rises by \(\varepsilon\), no deviation utility increases. Thus incentive compatibility and participation are preserved.
%Moreover, \(\varepsilon(\theta)\) is Borel whenever \(q_0(\theta)\) and \(q_1(\theta)\) are Borel, so the reportwise transformation remains an admissible Borel direct menu; it is one-sided by construction.

For each truthful type, the price rises by \(\varepsilon(\theta)\)
and the competitive cost by \(\tau\varepsilon(\theta)\). Integrating the
resulting typewise profit change gives
\eqref{eq:concentration-gain}. If \(\tau<1\), optimality therefore requires
\(\varepsilon(\theta)=0\) for \(F\)-almost every report. Applying the same
feasibility-preserving transformation to all reports then changes only a
null set and produces a profit-equivalent representative that is one-sided
everywhere. At \(\tau=1\), the transformation is profit neutral for every
menu.
\end{proof}

%\subsection{Two posterior types}

\begin{proof}[\textbf{Proof of Lemma~\ref{lem:canonical}}]
A one-sided experiment lies on one of two branches: \(q_1=1\) or \(q_0=1\). Solving \(\Delta(E,\theta_\ell)=L\alpha\) on each branch therefore gives all candidates.

For same-default types, the two candidates and their values to type \(h\) are
\begin{align*}
q_1=1:\quad
&(q_0,q_1)=(\alpha,1),
&V(E,\theta_h)&=H\alpha,\\
q_0=1:\quad
&(q_0,q_1)=\left(1,1-\frac{L(1-\alpha)}{\theta_\ell}\right),
&V(E,\theta_h)&=H\alpha+
\frac{(\theta_\ell-\theta_h)(1-\alpha)}{\theta_\ell}.
\end{align*}
Both candidates are obedient for type \(\ell\), and their probabilities lie in \([0,1]\). Because \(\theta_\ell>\theta_h\), the second candidate is weakly more valuable to type \(h\). Hence \(E(\alpha,1)\) is least attractive.

For opposite-default types, where \(L=\theta_\ell\), the candidates are
\begin{align*}
q_0=1:\quad
&(q_0,q_1)=(1,\alpha),
&V_h^\star(\alpha)&=
\max\{\theta_h\alpha+1-2\theta_h,0\},\\
q_1=1:\quad
&(q_0,q_1)=\left(
\frac{L\alpha+1-2L}{1-L},1\right),
&V_h^{\mathrm{alt}}(\alpha)&=
H\frac{L\alpha+1-2L}{1-L}.
\end{align*}
Again, both candidates are obedient for type \(\ell\), and their probabilities lie in \([0,1]\). The affine expression inside \(V_h^\star\) and \(V_h^{\mathrm{alt}}\) both equal \(H\) at \(\alpha=1\), while their slopes satisfy
\[
  \theta_h>
  \frac{HL}{1-L}
  \quad\Longleftrightarrow\quad
  \theta_h>L.
\]
Thus the first affine expression is weakly below \(V_h^{\mathrm{alt}}\) on \([0,1]\). Since \(V_h^{\mathrm{alt}}\geq0\), truncating the first expression at zero preserves this ordering. Hence \(E(1,\alpha)\) is least attractive.

Subtracting the low-type price \(L\alpha\) and imposing high-type individual rationality gives
\[
  r_h(\alpha)=
  \begin{cases}
    (H-L)\alpha,
      & \text{under \eqref{eq:same-default}},\\[0.3em]
    \max\{(\theta_h-L)\alpha+1-2\theta_h,0\},
      & \text{under \eqref{eq:opposite-default}}.
  \end{cases}
\]
The first expression lies in \([0,H-L]\). In the second, the affine term is increasing in \(\alpha\) and equals \(H-L\) at \(\alpha=1\). Therefore \(0\leq r_h(\alpha)\leq H-L\), proving \eqref{eq:rent-bound}.
\end{proof}

\begin{proof}[\textbf{Proof of Proposition~\ref{prop:no-distortion}}]
%First, the two-type problem attains its supremum. By Proposition~\ref{prop:responsive}, it suffices to consider two binary recommendation experiments, parameterized by four probabilities in
%\([0,1]\). Individual rationality implies, for
%\(k\in\{\ell,h\}\),
%\[
%  0\leq t_k\leq V(E_k,\theta_k)
%  \leq\min\{\theta_k,1-\theta_k\}.
%\]
%The obedience, incentive-compatibility, individual-rationality, and nonnegative-transfer inequalities define a nonempty closed subset of this compact finite-dimensional parameter set. The objective is continuous, so an optimal responsive menu exists; the responsive reduction makes it optimal in the original admissible class as well.

Take an optimal menu. Proposition~\ref{prop:responsive} gives a profit-equivalent responsive representative while preserving each intended type's value and transfer. Because both types have positive probability and \(\tau<1\), Proposition~\ref{prop:concentration} would strictly increase profit if its concentration increment were positive at either report. The responsive representative is therefore already one-sided at both reports. Write the value assigned to type \(\ell\) as \(L\alpha\), and replace its item by \(E_\ell^\star(\alpha)\). Lemma~\ref{lem:canonical} implies that this replacement preserves type \(\ell\)'s value and competitive cost and weakly reduces type \(h\)'s deviation value, so feasibility and profit are preserved.

Let \(V_h=V(E_h,\theta_h)\), and let \(U_\ell,U_h\) be the truthful rents in this representative menu. Since \(t_\ell=L\alpha-U_\ell\), high-type incentive compatibility and individual rationality imply
\[
  U_h
  \geq
  U_\ell+V\bigl(E_\ell^\star(\alpha),\theta_h\bigr)-L\alpha,
  \qquad
  U_h\geq0.
\]
Together with \(U_\ell\geq0\), these inequalities yield \(U_h\geq r_h(\alpha)\). Hence
\[
  \Pi(\mathcal M)
  \leq
  \gamma\bigl[(1-\tau)V_h-r_h(\alpha)\bigr]
  +(1-\gamma)(1-\tau)L\alpha.
\]

Now allocate \(E_\ell^\star(\alpha)\) at price \(L\alpha\) to type \(\ell\), and full information at price \(H-r_h(\alpha)\) to type \(h\). This menu gives rents \(0\) and \(r_h(\alpha)\). It is feasible because the definition of \(r_h\) makes the high-type constraint hold, while \eqref{eq:rent-bound} gives
\[
  L-\bigl[H-r_h(\alpha)\bigr]\leq0
\]
for the low-type deviation and also ensures nonnegative prices. Its profit is
\[
  \gamma\bigl[(1-\tau)H-r_h(\alpha)\bigr]
  +(1-\gamma)(1-\tau)L\alpha.
\]
It therefore improves on \(\mathcal M\) by at least
\[
  \gamma(1-\tau)(H-V_h).
\]
Because \(\gamma>0\), \(1-\tau>0\), and \(V_h\leq H\), optimality requires \(V_h=H\). For an interior type, this equality means that the buyer chooses the correct action in both positive-probability states, so its signal supports are disjoint across states and the assigned experiment is equivalent to full information. The responsive reduction preserved the original intended value, so the conclusion applies to the original optimal menu as well.
\end{proof}

\begin{proof}[\textbf{Proof of Proposition~\ref{prop:same-default}}]
For same-default types, Lemma~\ref{lem:canonical} gives
\[
  r_h(\alpha)
  =
  (H-L)\alpha
  =
  (\theta_\ell-\theta_h)\alpha.
\]
Transfers \(t_\ell=L\alpha\) and \(t_h=H-r_h(\alpha)\) give \eqref{eq:same-default-transfers}, and substitution into \eqref{eq:two-type-reduced-profit} gives \eqref{eq:same-default-profit}. The latter is affine in \(\alpha\), with slope
\[
  (1-\gamma)(1-\tau)L-\gamma(H-L).
\]
Because \(H-\tau L>0\), this slope is nonnegative exactly when
\[
  \gamma\leq
  \frac{(1-\tau)L}{H-\tau L}
  =\gamma_{\mathrm{same}}(\tau).
\]
The optimum is therefore \(\alpha=1\) below the cutoff, \(\alpha=0\) above it, and every \(\alpha\in[0,1]\) at equality. The transfers at the two endpoints are those stated in the proposition.
\end{proof}

\begin{proof}[\textbf{Proof of Proposition~\ref{prop:opposite-default}}]
For opposite-default types, Lemma~\ref{lem:canonical} gives
\[
  r_h(\alpha)
  =
  \max\{(\theta_h-L)\alpha+1-2\theta_h,0\}.
\]
The affine term crosses zero at \(\bar\alpha\) in \eqref{eq:zbar}; the condition \(H>L\) implies \(\bar\alpha\in(0,1)\). On \([0,\bar\alpha]\), the required rent is zero, so \eqref{eq:two-type-reduced-profit} has the strictly positive slope \((1-\gamma)(1-\tau)L\). No \(\alpha<\bar\alpha\) is therefore optimal.

On \([\bar\alpha,1]\), the required rent is affine, and the slope of profit is
\[
  (1-\gamma)(1-\tau)L-\gamma(\theta_h-L).
\]
Because \(\theta_h-\tau L>0\), this slope is nonnegative exactly when
\[
  \gamma\leq
  \frac{(1-\tau)L}{\theta_h-\tau L}
  =
  \gamma_{\mathrm{opp}}(\tau).
\]
Thus the optimum is \(\alpha=1\) below the cutoff, \(\alpha=\bar\alpha\) above it, and every \(\alpha\in[\bar\alpha,1]\) at equality. For \(\alpha\geq\bar\alpha\), the transfers are
\[
  t_\ell(\alpha)=L\alpha,
  \qquad
  t_h(\alpha)=H-r_h(\alpha)
  =\theta_h-(\theta_h-L)\alpha.
\]
At \(\alpha=1\), both prices equal \(L\); at \(\alpha=\bar\alpha\), they are \(t_h=H\) and \(t_\ell=L\bar\alpha\), as claimed.
\end{proof}

%\subsection{Continuum results}
\begin{proof}[\textbf{Proof of Proposition~\ref{prop:implementability}}]
If \((I,t)\) is responsively implementable, obedience gives \(U(\theta)=\Delta(I(\theta),\theta)-t(\theta)\). Applying incentive compatibility in both directions to \(\theta<\theta'\), and using \eqref{eq:I-value}, yields
\[
 (\theta'-\theta)I(\theta)
 \leq U(\theta')-U(\theta)
       -\bigl[1-\bar u(\theta')-(1-\bar u(\theta))\bigr]
 \leq(\theta'-\theta)I(\theta').
\]
Hence \(I\) is nondecreasing. 
By the standard envelope theorem for one-dimensional
incentive-compatible mechanisms \citep{milgrom2002envelope}, truthful utility \(U(\cdot)\) is absolutely continuous and satisfies, almost everywhere,
\[
U'(\theta)
=
\frac{\partial \Delta(I(\theta),\theta)}{\partial\theta}
=
-\bar u'(\theta)+I(\theta).
\]
Integrating from \(0\) to \(\theta\), and using
\(\bar u(0)=1\), gives
\begin{equation}
U(\theta)
=
U(0)+1-\bar u(\theta)
+\int_0^\theta I(z)\,dz.
\label{eq:continuum-envelope}
\end{equation}
Types 0 and 1 cannot gain from information; individual rationality and nonnegative transfers therefore imply \(U(0)=U(1)=0\). Evaluating \eqref{eq:continuum-envelope} at 1 gives \(\int_0^1I=0\), and substituting the envelope into \(t=\Delta-U\) gives \eqref{eq:continuum-transfer}. This proves necessity.

Conversely, let \(I\in\Iset\), define \(t\) by \eqref{eq:continuum-transfer}, and set
\[
  U(\theta)=1-\bar u(\theta)+\int_0^\theta I(z)\,dz.
\]
If \(\theta\leq1/2\), then \(U(\theta)\geq\theta-\theta=0\) because \(I\geq-1\). If \(\theta\geq1/2\), balance gives
\(U(\theta)=1-\theta-\int_\theta^1I(z)\,dz\geq0\) because \(I\leq1\). Moreover,
\[
  t(\theta)=
  \begin{cases}
    \displaystyle\int_\theta^1[I(z)-I(\theta)]\,dz,
      & I(\theta)\geq0,\\[2mm]
    \displaystyle\int_0^\theta[I(\theta)-I(z)]\,dz,
      & I(\theta)<0,
  \end{cases}
\]
so monotonicity makes \(t\geq0\). Since \(\Delta(I(\theta),\theta)=U(\theta)+t(\theta)\geq0\) and \(q_0+q_1=2-|I(\theta)|\geq1\), Lemma~\ref{lem:deviation} implies obedience.

Finally, for any report \(\hat\theta\), direct substitution gives
\[
  U(\theta)-\bigl[\Delta(I(\hat\theta),\theta)-t(\hat\theta)\bigr]
  =\int_{\hat\theta}^{\theta}[I(z)-I(\hat\theta)]\,dz\geq0,
\]
where monotonicity handles either ordering of \(\theta\) and \(\hat\theta\). Therefore, a deviation followed by obedience is unprofitable. Lemma~\ref{lem:deviation} says that ignoring the recommendations instead gives incremental utility \(-t(\hat\theta)\leq0\leq U(\theta)\). Hence no report-and-action deviation is profitable, completing the proof.
\end{proof}

\paragraph{\emph{Derivation of the virtual objective}}
For completeness, integrate \eqref{eq:continuum-transfer} under \(F\):
\begin{align*}
  \int_0^1 t(\theta)f(\theta)\,d\theta
  &={\int_0^1 I(\theta)\bigl(\theta-\ind\{I(\theta)\geq0\}\bigr)f(\theta)\,d\theta}
    -\int_0^1\int_0^\theta I(z)\,d z\, f(\theta)\,d\theta\\
  &=\int_0^1 I(\theta)\bigl(\theta-\ind\{I(\theta)\geq0\}\bigr)f(\theta)\,d\theta
    -\int_0^1 I(\theta)(1-F(\theta))\,d\theta.
\end{align*}
Since \(\int_0^1I=0\), the last term equals \(\int_0^1I(\theta)F(\theta)\,d\theta\). Subtracting \(\tau\int\Delta(I(\theta),\theta)f(\theta)\,d\theta\), using \eqref{eq:I-value}, and dropping the allocation-independent constant \(-\tau\E_F[1-\bar u(\theta)]\) yields \eqref{eq:virtual-objective}--\eqref{eq:phi-plus}.

\begin{proof}[\textbf{Proof of Proposition~\ref{prop:regular-continuum}}]
Define
\begin{equation*}
  \delta(c)=\Phi_-(c)-\Phi_+(1-c).
\end{equation*}
Regularity makes \(\delta\) strictly increasing, while
\[
  \delta(0)=-1,
  \qquad
  \delta(1/2)=(1-\tau)f(1/2)>0.
\]
Thus there is a unique \(c\in(0,1/2)\) satisfying \eqref{eq:cutoff}.

Set \(\lambda^{\star}=\Phi_-(c)=\Phi_+(1-c)\).
For a type \(\theta\), the balance-adjusted objective is
\begin{equation*}
  \mathcal{J}(I,\theta)-\lambda^{\star} I
  =\begin{cases}
    I\bigl(\Phi_-(\theta)-\lambda^{\star}\bigr), & I\leq0,\\
    I\bigl(\Phi_+(\theta)-\lambda^{\star}\bigr), & I\geq0.
  \end{cases}
\end{equation*}
The gap \eqref{eq:phi-gap} makes this function concave in \(I\), and strict monotonicity of the virtual values makes its maximizer equal to \(-1\), \(0\), and \(1\) on the three intervals in \eqref{eq:regular-allocation}. The allocation is nondecreasing and satisfies
\begin{equation*}
  \int_0^1 I^{\star}(\theta)\,d\theta=-c+c=0.
\end{equation*}

For any other feasible \(I\), the multiplier term integrates to zero, while pointwise maximization gives
\begin{equation*}
  \int_0^1\mathcal{J}(I(\theta),\theta)\,d\theta
  \leq \int_0^1\mathcal{J}(I^{\star}(\theta),\theta)\,d\theta.
\end{equation*}

It thus maximizes the original objective: the multiplier integrates to zero for every balanced allocation, while \(I^{\star}\) maximizes the adjusted objective pointwise. Substitution into \eqref{eq:continuum-transfer} gives price \(c\) on the served interval and zero outside it. The pointwise allocation is unique except at the two cutoff types, which have zero probability. The envelope formula then pins down the induced transfer almost everywhere.
\end{proof}

\begin{proof}[\textbf{Proof of Corollary~\ref{cor:uniform}}]
For \(F(\theta)=\theta\) and \(f(\theta)=1\),
\begin{equation*}
  \Phi_-(\theta)=(2-\tau)\theta,
  \qquad
  \Phi_+(\theta)=(2-\tau)\theta-(1-\tau).
\end{equation*}
\eqref{eq:cutoff} gives \(2(2-\tau)c=1\), proving \eqref{eq:uniform-price} and \eqref{eq:uniform-interval}.  Direct integration of price minus competitive cost over the served interval gives
\begin{equation*}
  c(1-2c)-2\tau\int_c^{1/2}\theta\,d\theta
  =\frac{(1-\tau)^2}{4(2-\tau)},
\end{equation*}
which proves \eqref{eq:uniform-profit}.
\end{proof}

\begin{proof}[\textbf{Proof of Proposition~\ref{prop:irregular}}]
By \eqref{eq:phi-gap}, \(\Phi_-\geq\Phi_+\), so \(\mathcal{J}(\cdot,\theta)\) is concave for every \(\theta\). For \(m\in[-1,1]\), define
\[
 W(m)=\sup\left\{
   \int_0^1\mathcal{J}(I(\theta),\theta)\,d\theta:
   I\in\Iset^\uparrow,\ B(I)=m
 \right\}.
\]
The constraint set is nonempty for every \(m\), since the constant allocation \(I\equiv m\) is feasible. Because allocations are bounded and monotone, and both the objective and the balance condition are continuous, the supremum defining \(W(m\)) is finite and attained. In particular, the balanced problem attains \(W(0)\).

Convex combinations of monotone allocations remain monotone, and concavity of \(\mathcal{J}\) thus makes \(W\) concave, meaning there is some
\(\lambda^{\star}\in\mathbb R\) such that
\[
  W(m)\leq W(0)+\lambda^{\star}m,
\]
for every \(m\in[-1,1]\).
If \(I^{\star}\) attains \(W(0)\), then for every \(I\in\Iset^\uparrow\),
\[
 \mathcal L_{\lambda^{\star}}(I)
 \leq W(B(I))-\lambda^{\star}B(I)
 \leq W(0)
 =\mathcal L_{\lambda^{\star}}(I^{\star}).
\]
Thus, every balanced optimum is supported by such a multiplier. Conversely, if a balanced \(I^{\star}\) maximizes \(\mathcal L_{\lambda^{\star}}\), then comparison with any other balanced allocation cancels the multiplier term and proves that \(I^{\star}\) maximizes the original objective. This establishes both directions and the existence of a primal--dual pair.
\end{proof}

\begin{proof}[\textbf{Proof of Proposition~\ref{prop:ironing}}]
Fix \(\lambda\in\mathbb R\) and \(I\in\Iset^\uparrow\), and let
\(I^-=\min\{I,0\}\) and \(I^+=\max\{I,0\}\). Both are
nondecreasing. Applying the standard convex-minorant ironing inequality
of \citet{toikka2011ironing} to the two linear branches gives
\[
 \mathcal L_\lambda(I)
 \leq
 \int_0^1\Bigl[
   (\widetilde\Phi_-(\theta)-\lambda)I^-(\theta)
   +(\widetilde\Phi_+(\theta)-\lambda)I^+(\theta)
 \Bigr]d\theta.
\]
Because \(\widetilde\Phi_-\geq\widetilde\Phi_+\), the
smallest and largest pointwise maximizers of the ironed integrand are
exactly \(\underline I_\lambda\) and
\(\overline I_\lambda\), and both are nondecreasing. The ironed
objective of any allocation is bounded by the value attained by these
pointwise selections. Their negative and positive parts are constant
on every corresponding ironing interval, so the ironing bound is an
equality for both. Therefore both maximize \(\mathcal L_\lambda\). If another allocation
also maximizes it, the ironing and pointwise bounds must both bind; it
must therefore maximize the ironed integrand pointwise almost
everywhere, which gives
\(\underline I_\lambda\leq I\leq\overline I_\lambda\) almost
everywhere.

By Proposition~\ref{prop:irregular}, some balanced \(\widehat I\) maximizes \(\mathcal L_{\lambda^{\star}}\) for a multiplier \(\lambda^{\star}\). The preceding bounds and \(B(\widehat I)=0\) imply
\[
 B(\underline I_{\lambda^{\star}})\leq0\leq B(\overline I_{\lambda^{\star}}).
\]
If the residuals differ, the coefficient \(\beta\) stated in the proposition belongs to \([0,1]\) and makes
\(I^{\star}=(1-\beta)\underline I_{\lambda^{\star}}+\beta\overline I_{\lambda^{\star}}\) balanced. The monotone feasible set is convex, and both endpoints maximize the concave functional \(\mathcal L_{\lambda^{\star}}\), and so their convex combination does as well. Proposition~\ref{prop:irregular} then makes \(I^{\star}\) optimal for the balanced problem. If the residuals coincide, their common value must be zero, so each extremal allocation is already balanced and optimal.
\end{proof}

% Keep references.bib in the same folder as this .tex file.
% plainnat is a portable author--date style. Replace it with a required .bst
% file if your journal or institution specifies a particular Harvard variant.
\bibliographystyle{plainnat}
\bibliography{references}

\end{document}